\documentclass[11pt]{article}

\usepackage{fullpage}

\usepackage{epigraph}
\usepackage{nicefrac}
\usepackage[margin=1in]{geometry}
\usepackage{relsize}

\usepackage{dsfont}

\usepackage{latexsym}

\usepackage[dvipsnames,usenames]{xcolor}
\usepackage{parskip}
\usepackage{color}
\usepackage{float}
\usepackage{bbm}
\floatstyle{boxed}
\newfloat{algorithm}{t}{lop}
\usepackage{tikz}
\usepackage{varwidth}

\usepackage{float}
\usepackage[titletoc,title]{appendix}
\usepackage[classfont=sanserif,langfont=roman,funcfont=italic]{complexity}
\usepackage{caption}
\usepackage{comment}
\usepackage{enumerate}
\usepackage{amsmath, amsthm, amssymb, amstext,  graphicx,amsopn} 
\usepackage{subfigure}
\usepackage{braket}

\usepackage{booktabs} 
\DeclareMathOperator{\Ex}{\mathbb{E}}

\usepackage{epigraph}
\usepackage{nicefrac}
\usepackage{hyperref}
\hypersetup{colorlinks=true,linkcolor=black,citecolor=blue}

\usepackage{mathtools,extarrows}
\usepackage{url}
\newtheorem{definition}{Definition}

\newtheorem{observation}{Observation}

\newtheorem{lemma}{Lemma}

\newtheorem{remark}{Remark}
\allowdisplaybreaks

\allowdisplaybreaks

\renewcommand{\vec}[1]{\mathbf{#1}}

\usepackage{thmtools} 
\usepackage{thm-restate}
\usepackage[capitalise,nameinlink]{cleveref}

\newtheorem{mylemma}[lemma]{Lemma}

\newtheorem{mydefinition}[definition]{Definition}

\crefname{figure}{Figure}{Figure}

\usepackage{mathrsfs}
\hypersetup{colorlinks=true,linkcolor=blue,citecolor=blue}
\definecolor{mygreen}{rgb}{0.0, 0.55, 0.0}
\definecolor{blue-violet}{rgb}{0.54, 0.17, 0.89}

\usepackage{tikz}
\usetikzlibrary{calc, graphs, graphs.standard, shapes, arrows, arrows.meta, positioning, decorations.pathreplacing, decorations.markings, decorations.pathmorphing, fit, matrix, patterns, shapes.misc, tikzmark}

\usepackage{amsfonts}

\newcommand{\eps}{\varepsilon}

\renewcommand{\R}{\mathbb{R}}

\newcommand{\cO}{{O}}
\newcommand{\cH}{\mathcal{H}}
\renewcommand{\cP}{\mathcal{P}}
\newcommand{\cX}{\mathcal{X}}
\newcommand{\cW}{\mathcal{W}}
\newcommand{\cY}{\mathcal{Y}}
\newcommand{\cZ}{\mathcal{Z}}

\newcommand{\X}{\vec{x}}
\newcommand{\Y}{\vec{y}}

\newcommand{\explain}[1]{\tag{\textcolor{gray}{#1}}}

\title{Spectral Lower Bounds for Local Search\footnote{Alphabetical author ordering.  Research  supported  by the US National Science Foundation under CAREER grant CCF-2238372.}}
\author{Simina Br\^anzei\footnote{Purdue University. E-mail: simina.branzei@gmail.com.}
	\and 
	Nicholas J. Recker\footnote{Purdue University. E-mail: nrecker@purdue.edu.}
}

\date{}

\begin{document}

	\maketitle

	\begin{abstract}
		
		Local search  is a powerful heuristic in optimization and computer science, the  complexity of which has been studied in the white box and black box models. In the black box model, we are given a graph $G = (V,E)$ and oracle access to a function $f : V \to \mathbb{R}$. The local search problem is to find a vertex $v$ that is a local minimum, i.e. with $f(v) \leq f(u)$ for all $(u,v) \in E$, using as few queries to the oracle as possible.

		We show that if a graph $G$ admits a lazy, irreducible, and reversible Markov chain  with stationary distribution $\pi$, then the randomized query complexity of local search on $G$ is $\Omega\left( \frac{\sqrt{n}}{t_{mix} \cdot \exp(3\sigma)}\right)$,  where  $t_{mix}$ is the  mixing time of the chain and $\sigma = \max_{u,v \in V(G)} \frac{\pi(v)}{\pi(u)}.$ This theorem formally establishes a connection between the  query complexity of local search and the mixing time of the fastest mixing Markov chain for the given graph.
		We also get several corollaries that lower bound the complexity as a function of the spectral gap, one of which slightly improves a prior result from \cite{BCR23}.
	\end{abstract}
	
	\maketitle 
	
	
	\newpage

	\section{Introduction}
	
	Local search stands as a robust heuristic within optimization and computer science, analyzed through both white box and black box frameworks. In the  black box model, we are given  a graph $G=(V,E)$ alongside oracle access to a function $f : V \to R$. The objective is to identify a vertex $v$ that represents a local minimum, meaning $f(v) \leq f(u)$  for every edge $(u,v)$, while minimizing the number of vertices queried.

	Obtaining lower bounds for the complexity of  local search has a rich history of analysis via random walks.
	The first  pioneering work on the subject was  \cite{aldous1983minimization}, which did  careful tailored analysis of the hitting time of random walks  on the Boolean hypercube to obtain lower bounds for local search.
	Another breakthrough was obtained by \cite{Aaronson06}, which designed a combinatorial method of obtaining lower bounds for local search inspired by the relational adversary method from quantum computing. This approach enabled obtaining sharper lower bounds for the Boolean hypercube and $d$-dimensional  grid, and was successfully used in many later works.
	
	In this paper we consider the high level question: How does the  geometry of the graph affect the complexity of local search?
	While the query complexity is comprehensively understood for neighbourhood structures such as the $d$-dimensional grid  and the Boolean hypercube, knowledge remains  limited for more general neighbourhood structures.
	
	Nevertheless, the spatial structure in optimization settings typically extends to more complex graphs. For instance, in scenarios such as low rank matrix estimation with data compromised by adversarial attacks,  the function is defined on a Riemannian manifold rather than a  traditional Euclidean space~\cite{bonnabel2013stochastic}; thus the discretization of an optimization search space may not necessarily always correspond to some $d$-dimensional grid.  For a more extensive survey on stochastic gradient descent on Riemannian  manifolds, see, e.g., \cite{amari_book}. This motivates studying  local search not only on hypercubes and grids, but also on  broader classes of graphs.

	Inspired by the observation that many lower bounds for local search are based on various types of random walks, we consider general random walks for the graph at hand and obtain lower bounds as a function of their mixing time. Our analysis uses a variant of the classical relational adversary method from \cite{BCR23} and our  main result is generic in two ways: the graph is  arbitrary and the random walk evolves according to a Markov chain, which we only require to be lazy, irreducible, and reversible. This allows us to formally connect the query complexity of local search and the mixing time of the fastest mixing Markov chain for the given graph, which is a  classical problem analyzed starting with \cite{boyd2004fastest}, with recent results in \cite{olesker2021geometric}.
	As a corollary, we also get a lower bound  in terms of the spectral gap of the transition matrix of the chain.

	

	\section{Our Results} \label{sec:model}
	
	In this section we define the model and state our results.
	
	\subsection{Model} 
	Let $G = (V,E)$ be a connected undirected  graph and $f : V \to \mathbb{R}$ a function defined on the vertices.  
	A vertex $v \in V$  is a local minimum if $f(v) \leq f(u)$ for all $\{u,v\} \in E$.
	We will write $V = [n] = \{1, \ldots, n\}$. 
	
	Given as input a graph $G$ and oracle access to function $f$, the local search problem is to find a local minimum of $f$ on $G$ using as few queries as possible. {Each query is of the form: ``Given a vertex $v$, what is $f(v)$?''}.

	\paragraph{Query complexity.} The \emph{deterministic query complexity} of a task is the total number of queries necessary and sufficient for a correct deterministic algorithm to return a solution.
	The \emph{randomized query complexity} is the expected number of queries required to return a solution with probability at least $9/10$ for each input, where the expectation is taken over the coin tosses of the protocol. 
	
	\paragraph{Degree and distance.} Let $d_{max}$ and $d_{min}$ be the maximum and minimum degree of any vertex in $G$ respectively.
	Let $d(v)$ be the degree of $v$ for all $v \in V$.
	For each $u,v \in V$, let $dist(u,v)$ be the length of the shortest path from $u$ to $v$. 
	\paragraph{Markov chain.}
	We consider a discrete-time Markov chain on $G$ with transition matrix $\cP$, meaning that the state space is $V$ and $\cP_{u,v} = \cP_{v,u} = 0 $ whenever $(u,v) \not \in E(G) \bigcup_{u \in V} \{\{u,u\}\}$.
	Suppose  the chain has stationary distribution $\pi$.
	The chain is:
	\begin{description}
		\item[$\;\;\bullet$] \emph{lazy}: if  $\cP_{u,u} \geq 1/2$ for all $u \in V$.
		\item[$\;\;\bullet$]  \emph{irreducible}:  if all states can be reached
		from any starting point. Formally, for any two states $x, y \in V$ there exists
		an integer $t$ (possibly depending on $x$ and $y$) such that $(\cP^t)_{x, y} > 0$, where  $\cP^t$  is the $t$-th power of the matrix $\cP$.
		\item[$\;\;\bullet$]  \emph{reversible}: if  $\pi(u) \cP_{u,v} = \pi(v) \cP_{v,u}$  for all $u,v \in V$. \footnote{For a formal definition of reversibility, see  \cite{levin2017markov} equation 3.26.}
	\end{description}

	For each $\epsilon > 0$, the \emph{mixing time} $t_{mix}(\epsilon)$ of the Markov chain\footnote{For a formal definition of mixing time, see e.g. \cite{levin2017markov} equation 4.30. The definition in \cite{levin2017markov} equation 4.30 is based on the TV distance, but is  equivalent to the one here by  Proposition 4.2 in \cite{levin2017markov}.} with transition matrix $\cP$ is:
	\begin{align}
		t_{mix}(\epsilon) = \min \left\{t \in \mathbb{N} \; \Bigl\lvert \; \forall u \in V:\;\; \frac{1}{2} \sum_{v \in V} |(\cP^t)_{u,v}-\pi(v)| \le \epsilon \right\} \,.
	\end{align}
	

	\subsection{Our Contributions}
	
	Our main contribution is the following theorem. 
	\begin{restatable}{mytheorem}{lowmixingtimeimplieslocalsearchhard} \label{thm:lower_bound_in_terms_of_mixing_time}
		Let $G = (V,E)$ be a connected undirected graph on $n$ vertices.
		Consider a discrete-time, lazy, irreducible, and  reversible  Markov chain on $G$ with transition matrix $\cP$ and stationary distribution $\pi$.
		%
		Then the  randomized query complexity of local search on $G$ is  
		\[        
		\Omega\left( \frac{\sqrt{n}}{t_{mix}\left(\frac{\sigma}{2n}\right)    \cdot \exp(3\sigma)}\right), \qquad \; \mbox{where} \; \;  \sigma = \max_{u,v \in V} \frac{\pi(v)}{ \pi(u)}\,.
		\]
	\end{restatable}
	
	The best lower bound given by \cref{thm:lower_bound_in_terms_of_mixing_time} is  attained by considering the Markov chain with the fastest possible mixing time for $G$; see \cite{boyd2004fastest} for a classical reference on this problem.
	
	Indeed, for many classes of graphs there can be significant gaps between the mixing time of the fastest mixing Markov chain and that of more obvious choices of Markov chains (such as the max-degree walk or the Metropolis-Hastings chain), including  the barbell graph, edge-transitive graphs, and distance transitive graphs~\cite{BDPX09}. 
	For example, when the stationary distribution is set to uniform on the barbell graph\footnote{The barbell graph consists of two cliques of $n/2$ vertices each, connected by a single edge.}, the max-degree random walk has mixing time $\Theta(n^3)$ while the fastest mixing walk mixes in only $\Theta(n^2)$ steps (see corollaries $5.2$ and $5.3$ in \cite{boyd2005symmetry}).

    Do not be alarmed by the exponential dependence on $\sigma$. There are natural choices for $\cP$ for which $\sigma = 1$, such as the max-degree walk. Writing \cref{thm:lower_bound_in_terms_of_mixing_time} in this way, rather than enforcing $\sigma = 1$, gives a little flexibility for small (ideally constant) $\sigma$.
	
	
	\begin{remark}
		Since  $\sigma \geq 1$, we always have  $t_{mix}(\sigma/(2n)) \le t_{mix}(1/(2n))$. Thus  \cref{thm:lower_bound_in_terms_of_mixing_time}  implies the randomized query complexity is
		\[ \Omega\left( \frac{\sqrt{n}}{t_{mix}\left(\frac{1}{2n}\right)  \exp(3\sigma)}\right)\,.
		\]
	\end{remark}

	
	For lazy chains, the second eigenvalue is always non-negative, so Theorem~\ref{thm:lower_bound_in_terms_of_mixing_time} implies a lower bound based on the spectral gap.
	This leads to the following corollary:
	
	\begin{restatable}{mycorollary}{highspectralgapimplieslocalsearchhard} \label{cor:local_search_and_spectral_gap}
		%
		Let $G = (V,E)$ be a connected undirected graph on $n$ vertices.
		Consider a discrete-time, lazy, irreducible, and  reversible  Markov chain on $G$ with transition matrix $\cP$ and stationary distribution $\pi$. 
		The  randomized query complexity of local search on $G$ is 
		\[
		\Omega\left( \frac{(1-\lambda_2) \sqrt{n}}{\log(n) \exp(3\sigma)} \right),
		\]
		where $\lambda_2$ is the second eigenvalue of $\cP$ and $\sigma = \max_{u,v \in V} {\pi(v)}/{ \pi(u)}$.
	\end{restatable}





	
	If we constrain the ratio $d_{max}/d_{min}$ and focus on the simple lazy random walk, we can remove the dependency on $\sigma$ in \cref{cor:local_search_and_spectral_gap}.
	
	\begin{restatable}{mycorollary}{highspectralgapimplieslocalsearchhardboundedsigma} \label{cor:local_search_and_spectral_gap_bounded_sigma}
		Let $G = (V,E)$ be a connected undirected graph on $n$ vertices.
		If $d_{max}/d_{min} \leq C$ for some constant $C > 0$, then the randomized query complexity of local search on $G$ is
		\[ 
		\Omega\left( \frac{(1-\lambda_2) \sqrt{n}}{\log{n}} \right),
		\]
		where  $\lambda_2$ is the second eigenvalue of the transition matrix of the simple lazy random walk on $G$.
	\end{restatable}

	The lower bound in \cref{cor:local_search_and_spectral_gap_bounded_sigma} improves by a $\log{n}$ factor the lower bound attainable from \cite{BCR23} for such graphs. For comparison, we state the lower bound from  \cite{BCR23} next.
	
	\begin{restatable}{myproposition}{bcrcor}\cite{BCR23}\label{cor:local_search_and_spectral_gap_sharp_power_law}
		Let $G = (V,E)$ be a connected undirected graph on $n$ vertices.
		If $d_{max}/d_{min} \le C$ for some constant $C > 0$, then the randomized query complexity of local search on $G$ is  
		\[ 
		\Omega\left( \frac{(1 - \lambda_2) \sqrt{n}}{\log^2(n)} \right),
		\]
		where $\lambda_2$ is the second eigenvalue of
		the transition matrix of the simple lazy random walk on $G$.
	\end{restatable}
	
	The extra factor of $\log(n)$ in \cref{cor:local_search_and_spectral_gap_sharp_power_law} stems from the result used to connect expansion to edge congestion (see \cite{chuzhoy2016routing}, corollary C.2), which rounds  a fractional flow with congestion $n \log(n)$ to an integer flow with congestion $n \log^2(n)$.
	By avoiding expansion altogether, \cref{cor:local_search_and_spectral_gap_bounded_sigma} also avoids this excess factor of $\log(n)$.
	
	
	
	
	\subsection{High Level Approach}
	
	
	
	
	
	
	
	
	The high level approach is as follows. We consider a set of value functions $f : V \to \mathbbm{R}$ induced by walks from a fixed starting vertex. We define $f$ such that the value at any vertex off the walk is the distance to the starting vertex of the walk, and the value at vertices on the walk is decreasing along the walk. This ensures that there is only one local minimum, namely the end of the walk. This type of construction is classical \cite{aldous1983minimization,Aaronson06}. We denote the space of such functions $\cX$.
	
	We then choose a similarity measure $r : \cX \times \cX \to \mathbbm{R}_{\ge 0}$ (also called the ``relation'') between any two functions. Semantically, the relation measures the difficulty of distinguishing two possible input functions from each other; thus it will be useful that  functions induced by walks with a long shared prefix are defined as more related. Given $r$, we invoke the relational adversary variant (\cref{thm:relational_adversary} from \cite{BCR23}), which implicitly defines a distribution over inputs based on $r$, and outputs a  lower bound on the randomized query complexity. 
	We choose $r$ carefully such that the distribution over walks is that of an arbitrary Markov chain.
	
	The innovation of our methodology over \cite{BCR23} lies in the construction; where their construction is based on low-congestion paths, ours takes the more natural and general approach of using arbitrary Markov chains, including the usual lazy random walks.

	\begin{figure}[h]
		\begin{center}
			\includegraphics[scale=1.6]{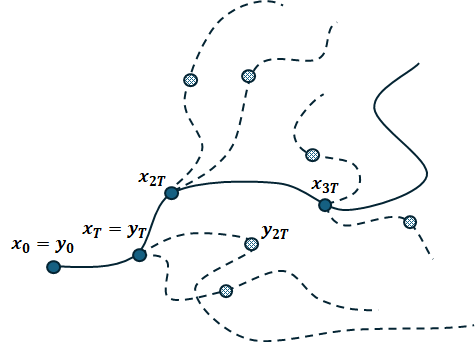}
			\caption{Consider a graph $G$ with a lazy, irreducible, and reversible Markov chain $\cP$ with stationary distribution $\pi$ and mixing time $T$. The proof fixes a walk $\vec{x} = [x_0, \ldots, x_L]$, where  $L = \lfloor \sqrt{n} \rfloor \cdot T$. 
				The walk $\vec{x}$ is illustrated as a solid line, where every $T$-th node is highlighted.
				Sample a random walk $\vec{y} = [y_0, \ldots, y_L]$ according to $\cP$, conditioned on $\vec{y}$ and $\vec{x}$ having a shared prefix of length $jT$, where $j \sim U(0,  \lfloor \sqrt{n} \rfloor )$. 
				We say that $x_{jT} = y_{jT}$ is the ``divergence point''.
				In the figure, the shared prefix of $\vec{x}$ and $\vec{y}$ is $[x_0, \ldots, x_T]$, so $j=1$. 
				A critical step of the proof is to show that no vertex $v$ is too likely to lie on $\vec{y}$ after its divergence  from $\vec{x}$. To show this, we divide the walk $\vec{y}$ in two regions.
				Vertices in the region $R_1 = [y_{jT}, \ldots, y_{jT+1}]$ are \emph{collectively} close to being distributed according to $\pi$.
				This is because the divergence point from $\vec{x}$ is chosen randomly.
				In the region $R_2 = [ y_{jT+1}, ..., y_{L}]$, the walk $\vec{y}$ has mixed, so the vertices in $R_2$ are close to being distributed according to $\pi$. 
				In either case, no vertex $v$ is too likely to lie on $\vec{y}$ after diverging from $\vec{x}$. 
			}
			\label{fig:mixing}
		\end{center}
	\end{figure}
	
	A key step in analyzing the formula given by the relational adversary is the following. When fixing a staircase $\vec{x}$ and sampling a second staircase $\vec{y}$ conditioned to share an initial prefix of random length with $\vec{x}$, the proof has to show that no vertex $v$ is too likely to lie on the ``tail'' of $\vec{y}$ (i.e. the portion after the initial segment shared with $\vec{x}$). The difference between our setting and  previous random walk based methods lies in this analysis. 
	
	We analyze separately the portions of the tail before and after it mixes. After $\vec{y}$ mixes, it is  close to being distributed according to the stationary distribution $\pi$. Before $\vec{y}$ mixes, the randomness of the point on $\vec{x}$ at which $\vec{y}$ diverges suffices to keep $\vec{y}$ close enough to being distributed according to $\pi$. A visual depiction is shown in  \cref{fig:mixing}. 
	
	This analysis is very generic, parameterized only by the stationary distribution of the walk used. This allows \cref{thm:lower_bound_in_terms_of_mixing_time} to give results for walks from arbitrary Markov chains with no additional analysis other than estimating the mixing time. The natural mixing properties of the lazy random walk on expanders then allow us to derive strong lower bounds for such graphs.

	\subsection{Related Work}
	
	\paragraph{The Boolean hypercube and the $d$-dimensional grid.} The query complexity of local search was first studied experimentally  by \cite{tovey81}.
	The first breakthrough in the theoretical analysis of local search was obtained by \cite{aldous1983minimization}.
	Aldous stated the algorithm based on  steepest descent with a warm start and showed the first nontrivial lower bound of $\Omega(2^{n/2-o(n)})$ on the query complexity for the Boolean hypercube $\{0,1\}^n$.
	
	The lower bound construction from \cite{aldous1983minimization} uses Yao's lemma and describes a hard distribution, such that if a deterministic algorithm receives a random function according to this distribution, the expected number of queries 
	issued until finding a local minimum will be large. The random function is obtained as follows:
	
	\begin{quote} Consider an initial vertex $v_0$ uniformly at random. Set the function value at $v_0$ to $f(v_0) = 0$. From this vertex, start an unbiased random walk $v_0, v_1, \ldots$ For each vertex $v$ in the graph, set $f(v)$ equal to the first hitting time of the walk at $v$; that is, let $f(v) = \min\{t \mid v_t=v\}$.
	\end{quote}
	The function $f$ defined this way has a unique local minimum at $v_0$.
	By a very delicate  analysis of this distribution,  \cite{aldous1983minimization} showed a lower bound of $\Omega(2^{n/2-o(n)})$ on the hypercube $\{0,1\}^n$.
	This almost matches the query complexity of steepest descent with a warm start, which was also analyzed in \cite{aldous1983minimization} and shown to take $\cO(\sqrt{n} \cdot 2^{n/2})$ queries in expectation on the hypercube.
	The steepest descent with a warm start algorithm applies to generic graphs too,
	resulting in $\cO(\sqrt{n \cdot d_{max}})$ queries overall for any graph with maximum degree $d_{max}$.
	
	Aldous' lower bound for the hypercube was later improved by 
	\cite{Aaronson06} to $\Omega(2^{n/2}/n^2)$ via a relational adversary method inspired from quantum computing.
	\cite{zhang2009tight} further improved this lower bound to  $\Theta(2^{n/2} \cdot \sqrt{n})$ via a ``clock''-based random walk construction, which avoids self-intersections.
	Meanwhile, \cite{llewellyn1989local} developed a deterministic divide-and-conquer approach to solving local search that is theoretically optimal over all graphs for deterministic algorithms.
	On the hypercube, their method yields a lower bound of $\Omega(2^n/\sqrt{n})$ and an upper bound of $\cO(2^n \log(n)/\sqrt{n})$.
	
	For the  $d$-dimensional grid $[n]^d$, 
	\cite{Aaronson06} used the relational adversary method there to show a randomized lower bound of $\Omega(n^{d/2-1} / \log n)$ for every constant $d \ge 3$.
    The relational adversary method that we use is inspired by Aaronson's classical relational adversary.
	\cite{zhang2009tight} proved a randomized lower bound of $\Omega(n^{d/2})$ for every constant $d \ge 4$; this is tight as shown by Aldous' generic upper bound.
	Zhang also showed improved bounds of $\Omega(n^{2/3})$ and $\Omega(n^{3/2} / \sqrt{\log n})$ for $d=2$ and $d=3$ respectively, as well as some quantum results.
	The work of \cite{sun2009quantum} closed further gaps in the quantum setting as well as the randomized $d=2$ case.
	
	The problem of local search on the grid was also studied under the context of multiple limited rounds of adaptive interactions by \cite{BL22}.  \cite{babichenko2019communication} studied the communication complexity of local search; this captures  distributed settings, where data is stored  on different computers.
	
	\paragraph{General graphs.}
	More general results are few and far between.
	On many graphs, the simple bound from \cite{aldous1983minimization} of $\Omega(d_{max})$ queries is the best known lower bound: hiding the local minimum in one of the $d_{max}$ leaves of a star subgraph requires checking about half the leaves in expectation to find it.
	\cite{santha2004quantum} gave a quantum lower bound of $\Omega\left( \sqrt[8]{\frac{s}{d_{max}}} / \log(n) \right)$, where $s$ is the separation number of the graph. This  implies the same lower bound in a randomized context, using the spectral method.
	Meanwhile, the best known upper bound is $\cO((s + d_{max}) \cdot \log n)$ due to \cite{santha2004quantum}, which was obtained via a refinement of the divide-and-conquer procedure of \cite{llewellyn1989local}. 
	
	\cite{dinh2010quantum} studied Cayley and vertex transitive graphs and gave lower bounds for local search as a function of the number of vertices and the diameter of the graph.
	\cite{Verhoeven06} obtained upper bounds as a function of the genus of the graph.
	\cite{BCR23} improved on the relational adversary method of \cite{Aaronson06} and applied it to get lower bounds for local search as a function of graph congestion, with corollaries for expansion and spectral gap. We apply their methodology to an alternate construction to achieve our results, including an improvement on their spectral gap result.

	\paragraph{Stationary points.} Local search is strongly related to the problem of local optimization where the goal is to find a stationary point of a function on $\mathbb{R}^d$.
	A standard method for solving local optimization problems is through gradient  methods, which find $\epsilon$-stationary points. 
	To show lower bounds for finding stationary points, one can similarly  define a function that selects a walk in the underlying space and hides a stationary point at the end of the walk.  Handling (potential) smoothness of the function and  ensuring the stationary point is unique  represent additional challenges. 
	For examples of works on algorithms and complexity of computing approximate stationary points, see, e.g., \cite{vavasis1993black,pmlr-v119-zhang20p,  stationary_I,stationary_II,bubeck2020trap,pmlr-v119-drori20a}).
	
	
	\paragraph{Complexity classes.}
	The computational complexity of local search is captured by the class PLS, defined by \cite{DBLP:journals/jcss/JohnsonPY88} to model the difficulty of finding locally optimal solutions to optimization problems.
	A related class is PPAD, introduced in \cite{Papadimitriou_1994} to study the computational complexity of finding a Brouwer fixed-point.
	Both PLS and PPAD are subsets of  the class TFNP.
	
	The class PPAD contains many natural problems that are computationally equivalent to finding a Brouwer fixed point (\cite{CD09}), such as finding an approximate Nash equilibrium in a multiplayer or two-player game (\cite{daskalakis2009complexity,chen2009settling}), an Arrow-Debreu equilibrium in a market (\cite{vazirani2011market,chen2017complexity}), and a local min-max point (\cite{costis_minmax}).
	The query complexity of computing an $\eps$-approximate Brouwer fixed point was studied in a series of papers starting with \cite{hirsch1989exponential},  which introduced a construction where  the function is induced by a  hidden walk. This was  later improved by \cite{chen2005algorithms} and \cite{chen2007paths}.
	Recently, \cite{fearnley2022complexity} showed that the class CLS, introduced by \cite{daskalakis2011continuous} to capture continuous local search, is equal to PPAD $\cap$ PLS.
	

	\section{Preliminaries} 
	We fix a discrete-time Markov chain with transition matrix $\cP$ that has the properties required by \cref{thm:lower_bound_in_terms_of_mixing_time}: lazy, irreducible, and reversible. Let $\pi$ denote the unique stationary distribution of the chain. For each $S \subseteq V$, let $\pi(S) = \sum_{v \in S} \pi(v)$. 
	Moreover, let  
	\begin{align}
		\sigma = \max_{u,v \in V} {\pi(v)}/{\pi(u)}\,.
	\end{align}
	For every $k \in \mathbbm{N}$ and every walk $\vec{x} = (x_0, x_1, \ldots, x_k)$  in $G$, let $\cP[\vec{x}]$ be the probability that the random walk started at $x_0$ with transition matrix $\cP$ has trajectory $\vec{x}$, that is:
	\begin{align} \label{eq:def:Pr_x}
		\cP[\vec{x}] = \prod_{i=0}^{k-1} \cP_{x_i, x_{i+1}}\,.
	\end{align}

	\paragraph{Bottleneck Ratio.}
	The bottleneck ratio $\Phi_\star$ of the Markov chain with transition matrix $\cP$ is \footnote{See, e.g.,  \cite{levin2017markov} equation 7.5.}:
	\[
	\Phi_\star = \min_{\substack{S \subseteq V:  \pi(S) \le \frac{1}{2}}} \; \sum_{u \in S,v \in V \setminus S} \frac{\pi(u) \cP_{u,v}}{\pi(S)}.
	\] 
	
	\paragraph{Visiting probability.}
	For each pair of vertices $u,v \in V$ and integer $\ell \in \mathbbm{N}$: 
	\begin{itemize}
		\item let  $\mbox{P}_{visit}(u,v,\ell)$ be the probability that a random walk that has  transition matrix  $\cP$,  length $\ell$, and starts at $u$  visits $v$.
		\item let $\mbox{E}_{visit}(u,v,\ell)$ be the expected number of times that a random walk that has transition matrix $\cP$,  length $\ell$, and starts at $u$  visits $v$.
		\item let  $\mbox{P}_{end}(u,v,\ell)$ be the probability that a random walk that has  transition matrix $\cP$,  length $\ell$, and starts at $u$  ends at $v$.
	\end{itemize}

	\paragraph{Edge Expansion.} Let $E(S, V\setminus S) = \{(u,v) \in E \mid u \in S, v \in V \setminus  S\}$ be the set of edges with one endpoint in $S$ and the other in $V \setminus S$.
	The edge expansion of $G$ is 
	\[ \beta = \min_{\substack{S \subseteq V: |S| \le n/2}} \frac{\left|E(S, V\setminus S) \right|}{|S|}\,.
	\]

	One of the main ingredients in our proof is a variant of the (classical) relational adversary method from quantum computing given in \cite{BCR23}.
	
	\begin{lemma}[\cite{BCR23}, Theorem 3]
		\label{thm:relational_adversary}
		Consider finite sets $A$ and $B$, a set $\cX \subseteq B^{A}$ of functions, and a map $\cH : \cX \to \{0,1\}$ which assigns a label to each function in $\cX$.
		Additionally, we get oracle access to an unknown function $F^* \in \cX$.
		The problem is to compute $\cH(F^*)$ using as few queries to $F^*$ as possible.\footnote{In other words, we have free access to $\cH$ and the only queries  counted are the ones to $F^*$, which will be of the form: ``What is $F^*(a)$?'', for some $a \in A$. The oracle will return $F^*(a)$ in one computational step.}
		
		Let $r : \cX \times \cX \to \R_{\geq 0}$ be a non-zero symmetric function of our choice with $r(F_1, F_2) = 0$ whenever $\cH(F_1) = \cH(F_2)$.
		For each $\cZ \subseteq \cX$, define  
		\begin{align}  \label{eq:def_M_and_q}
			M(\cZ) = \sum_{F_1 \in \cZ} \sum_{F_2 \in \cX} r(F_1, F_2) \;  \; \; \; \text{ and } \; \; \; \; q(\cZ) = \max_{a \in A}{\sum_{F_1 \in \cZ} \sum_{F_2 \in \cZ} r(F_1,F_2) \cdot \mathbbm{1}_{\{F_1(a) \neq F_2(a)\}}}\,.
		\end{align}
		
		If there exists a subset $\cZ \subseteq \cX$ with $q(\cZ) > 0$, 
		then the randomized query complexity of the problem is
		at least
		\begin{equation}
			\label{eq:relational_variant}
			\min\limits_{\substack{\cZ \subseteq \cX: q(\cZ) > 0}}
			0.01 \cdot  {M(\cZ)}/{q(\cZ)}  \,.
		\end{equation}
	\end{lemma}

	\section{Lower bounds via mixing time}

	To get lower bounds for local search, we will analyze the performance of deterministic algorithms when the input distribution is obtained by considering random functions, each of which is defined using a classical ``staircase'' construction. Each staircase is a random walk with transition  matrix $\cP$ on $G$. 
	We first give the setup and then prove the main theorem.
	
	\subsection{Setup}
	


	\begin{mydefinition}[Set of walks $\mathcal{W}$ and  parameter $T$] \label{def:W_and_T}
		Let $L = \lfloor \sqrt{n} \rfloor \cdot T$, where \mbox{$T = t_{mix}(\frac{\sigma}{2n})$}.
		Let $\cW$ be the set of walks $\{\vec{w} \mid \vec{w} = (w_0, \ldots, w_L) \}$ in $G$ with $w_0$ equal to the vertex $1$ and with $\cP_{w_i, w_{i+1}} > 0$ for all $0 \le i < L$.
	\end{mydefinition}
	
	\begin{mydefinition}[Milestones] \label{def:milestone}
		Given a walk $\vec{x} = (x_0, x_1, \ldots, x_L) \in \cW$,  every $T$-th vertex of the walk (including the first vertex) is called a ``milestone''.
		E.g.,  the first three  milestones of $\vec{x}$ are $x_0$, $x_T$, and $x_{2T}$.
	\end{mydefinition}
	
	\begin{mydefinition}[Good/bad walk] \label{def:good}
		A walk $\vec{x} \in \cW$ is ``good'' if it does not repeat any milestones and ``bad'' otherwise. 
		Let $good(\vec{x}) = True$ if  $\vec{x}$ is good and $False$ otherwise.
	\end{mydefinition}
	
	\begin{mydefinition}[Heads and Tails.] \label{def:heads_and_tails}
		For every walk $\vec{x} = (x_0, x_1,  \ldots, x_L) \in \cW$, let  
		\begin{align} 
			\begin{cases} 
				Head(\vec{x},j)  = (x_0, x_1,  \ldots, x_{j \cdot T})
				& \forall  j \in \{0, \ldots, \lfloor \sqrt{n} \rfloor\}  \,. \\
				Tail(\vec{x}, j) = (x_{j \cdot T+1}, x_{j \cdot T+2}, \ldots, x_L)
				& \forall  j \in \{0, \ldots, \lfloor \sqrt{n} \rfloor\}  \,. \\
				Tail(\vec{x},j_1,j_2) = (x_{j_1 \cdot T+1}, x_{j_1 \cdot T+2}, \ldots, x_{j_2 \cdot T})
				& \forall j_1, j_2 \in \{0, \ldots, \lfloor \sqrt{n} \rfloor \} \mbox{ with } j_1 \leq j_2\,.
			\end{cases}
		\end{align}
		For all $\vec{x}, \vec{y} \in \cW$,  define 
		$J(\vec{x},\vec{y})$ as the maximum index $j$ with $Head(\vec{x},j) = Head(\vec{y},j)$.
	\end{mydefinition}


	\begin{definition}[The functions $f_{\vec{x}}$ and $g_{\vec{x},b}$; the  set $\mathcal{X}$] \label{def:f_x_and_g_x_b}
		For each walk $\vec{x} = (x_0,  \ldots, x_L) \in \cW$, define a  function $f_{\vec{x}} : [n] \to \{-L, -L+1, \ldots, n\}$  such that for all $v \in [n]$:
		\begin{align}
			f_{\vec{x}}(v) = 
			\begin{cases}
				dist(v,1) & \text{ if } v \not \in \vec{x}  \\
				-\max\Bigl\{i \in \{0, \ldots, L\} \mid  x_i = v\Bigr\} & \text{ if } v \in \vec{x}\,. \\ 
			\end{cases}
		\end{align}
		For all $b \in \{0,1\}$, define  $g_{\vec{x},b} : [n] \to \{-L,-L+1, \ldots, n\} \times \{-1,0,1\}$ so that for all $ v \in [n]$:
		\begin{align}
			g_{\vec{x},b}(v) = 
			\begin{cases}
				(f_{\vec{x}}(v), -1) & \text{ if } v \neq x_{L} \\
				(f_{\vec{x}}(v), b) & \text{ if } v = x_{L}\,.
			\end{cases}
		\end{align}
		Let $\cX = \Bigl\{g_{\vec{x},b} \mid \vec{x} \in \cW \text{ and } b \in \{0,1\}\Bigr\}$.
	\end{definition}
	
	
	\begin{definition} [Valid function] \label{def:valid_function}
		Let $\vec{x} = (x_0, \ldots, x_\ell)$ be a walk in $G$. A function $f : V \to \mathbb{R}$ is valid with respect to the walk $\vec{x}$ if it satisfies the next conditions:
		\begin{enumerate}
			\item For all $u,v \in \vec{x}$, if $\max\{i \in \{0, \ldots, \ell\} \mid v = x_i\} < \max\{i \in \{ 0, \ldots, \ell\} \mid u = x_i\}$, then $f(v) > f(u)$.
			{In other words, as one walks along the walk $\vec{x}$ starting from $x_0$ until $x_\ell$, if the last time the vertex $v$ appears is before the last time that vertex $u$ appears, then $f(v) > f(u)$.}
			
			\item For all $v \in V \setminus \vec{x}$, we have $f(v) = dist(x_0, v) > 0$.
			
			\item $f(x_i) \leq 0$ for all $i \in \{0, \ldots, \ell \}$.
		\end{enumerate}
	\end{definition}

	
	We define a similarity measure $r$ between functions from $\cX$, commonly referred to as the \emph{relation}.

	\begin{mydefinition} [The function $r$]\label{def:r}
		Let $r : \mathcal{X} \times \mathcal{X} \to \mathbb{R}_{\ge 0}$  be a symmetric function such that for  each \mbox{$\X, \Y \in \cW$} and $b_1,b_2 \in \{0,1\}$,  
		\begin{align}
			r(g_{\X,b_1}, g_{\Y,b_2}) = \begin{cases}
				0 & \text{If } b_1 = b_2 \text{ or } \vec{x} = \vec{y} \text{ or } \X \text{ is bad or } \Y \text{ is bad.}\\
				\frac{\cP[\vec{x}] \cP[\vec{y}]}{\cP[Head(\vec{y},j)]} & \text{Otherwise, where $j = J(\vec{x},\vec{y})$}\,. 
			\end{cases} \notag 
		\end{align}
	\end{mydefinition}

    This value $r(\vec{x},\vec{y})$ may be interpreted as the probability that two walks sampled from the Markov chain will be $\vec{x}$ and $\vec{y}$ respectively, conditioned on diverging for the first time at an index in the range $(jT,(j+1)T]$.
	
	\begin{observation}
		The function $r$ from \cref{def:r} is symmetric,  since by definition of $j$ we have $Head(\vec{x},j) = Head(\vec{y},j)$.
	\end{observation}

	Having defined the relation $r$,  applying  
	\cref{thm:relational_adversary}
	will give a lower bound by considering a distribution $p$ over $\cX$, where each function $F \in \cX$ is given as input with probability 
	$
	p(F) = \frac{M(\{F\})}{M(\cX)},
	$	
	where $M$ is as defined in \cref{eq:def_M_and_q} for the relation $r$.
	
	What remains to be done is to explain how to invoke \cref{thm:relational_adversary} and estimate the lower bound it gives when the input  distribution is $p$.
	
	\subsection{Proof of the main theorem}


	
	\lowmixingtimeimplieslocalsearchhard* 
	\begin{proof}
		First, we may assume that $n \ge 16\sigma^2$.
		This is because in the alternative case where $n < 16\sigma^2$, the theorem statement does not give anything useful as 
		$\sqrt{n} < \exp(3\sigma)$ and $t_{mix}(\sigma/(2n)) \ge 1$.

		The value functions we will use are of the form $f_{\vec{x}}$ as seen in \cref{def:f_x_and_g_x_b}.
		These functions are parametrized by walks $\vec{x}$ of length $L$ from the set $\cW$ defined in \cref{def:W_and_T}.
		For sake of invoking \cref{thm:relational_adversary} however, we must turn the local search problem into a decision problem.
		To do this, we use the technique shown in ~\cite{Aaronson06,dinh2010quantum}:
		associate with each function $f_{\vec{x}}$ the function $g_{\vec{x},b}$ defined in \cref{def:f_x_and_g_x_b}.
		This hides a bit at the local minimum vertex (while hiding the value $-1$ at every other vertex).
		The decision problem is: given the graph $G$ and oracle access to a function $g_{\vec{x},b}$, return the hidden bit $b$; i.e. we set the function $\cH$ for use in \cref{thm:relational_adversary} as $\cH(g_{\vec{x},b}) = b$.
		
		
		By \cref{lem:proof_f_X_is_valid}, the function $f_{\vec{x}}$ as defined in \cref{def:f_x_and_g_x_b} is valid.
		Therefore by \cref{lem:valid_unique_local_min}, $f_\vec{x}$ has a unique local minimum, namely $x_L$.
		This means that $g_{\vec{x},b}$ as defined in \cref{def:f_x_and_g_x_b} does indeed hide the bit $b$ only at the local minimum of $f_\vec{x}$.
		Therefore measuring the query complexity of this decision problem will give the answer for local search, as the following two problems have query complexity within additive $1$:
		\begin{itemize}
			\item \emph{search problem}: given oracle access to a function $f_x$, return a vertex $v$ that is a local minimum;
			\item \emph{decision problem}: given oracle access to a function $g_{x,b}$, return $b$.
		\end{itemize}
		
		We then invoke \cref{thm:relational_adversary} with $\cX$ as defined in \cref{def:f_x_and_g_x_b} and with $\cH(g_{\vec{x},b}) = b$.
		This tells us that the randomized query complexity of the decision problem is at least 
		$\min\limits_{\substack{\cZ \subseteq \cX: q(\cZ) > 0}}
		0.01 \cdot  {M(\cZ)}/{q(\cZ)}$, with $M(\cZ)$ and $q(\cZ)$ as defined in \cref{thm:relational_adversary}.
		
		By \cref{lem:exists_Z_with_qZ_positive}, there exists a subset  $\cZ \subseteq \mathcal{X}$  with $q(\cZ) > 0$.
		From this point on, we fix an arbitrary subset $\cZ \subseteq \mathcal{X}$ of functions with $q(\cZ) > 0$ and will then lower bound $M(\cZ)$ and upper bound $q(\cZ)$.  
		
		\paragraph{Lower bounding $M(\cZ)$.}
		Because $q(\cZ) > 0$, we know $\cZ$ is not empty.
		Consider an arbitrary function $g_{\vec{x},b_1} \in \cZ$ with $\vec{x}$ good. 
		By definition of $M$, we have 
		\begin{align}
			M(\{g_{\vec{x},b_1}\}) & = \sum_{g_{\vec{y},b_2} \in \cX} r(g_{\vec{x},b_1}, g_{\vec{y},b_2}),
		\end{align}
		recalling:
		\begin{description}
			\item[$\bullet$] $\cP[\vec{y}] = \prod_{i=0}^{k-1} \cP_{y_i, y_{i+1}}$ for every walk $\vec{y} = (y_0, y_1, \ldots, y_k)$  in $G$;
			\item[$\bullet$] \cref{def:heads_and_tails} of $Head(\vec{y},j)$ and $J(\vec{x},\vec{y})$ and  \cref{def:r} of  the relation $r$:
			\begin{align}
				r(g_{\X,b_1}, g_{\Y,b_2}) = \begin{cases}
					0 & \text{If } b_1 = b_2 \text{ or } \vec{x} = \vec{y} \text{ or } \X \text{ is bad or } \Y \text{ is bad.}\\
					\frac{\cP[\vec{x}] \cP[\vec{y}]}{\cP[Head(\vec{y},j)]} & \text{Otherwise, where $j = J(\vec{x},\vec{y})$}\,. 
				\end{cases} \notag 
			\end{align}
		\end{description} 
		
		Then we can rewrite $M(g_{\vec{x}, b_1})$ as 
		\begin{align}
			M(\{g_{\vec{x},b_1}\})  & = \sum_{\substack{g_{\vec{y},b_2} \in \mathcal{X}:\\  b_2 = 1-b_1\\good(\vec{y})\\\vec{x} \neq \vec{y}}} \frac{\cP[\vec{x}] \cP[\vec{y}]}{\cP[Head(\vec{y},J(\vec{x},\vec{y}))]}   \,. \label{eq:rewriting_M_to_have_diff_bits}
		\end{align}
		
		The condition $\vec{x} \neq  \vec{y}$ in \cref{eq:rewriting_M_to_have_diff_bits} is equivalent to $J(\vec{x},\vec{y}) \neq \lfloor \sqrt{n} \rfloor$, so it must be the case that $J(\vec{x}, \vec{y}) \in \left\{ 0, \ldots, \lfloor \sqrt{n} \rfloor  - 1 \right\}$. We decompose the summation in \cref{eq:rewriting_M_to_have_diff_bits} by the value of $J(\vec{x}, \vec{y})$, excluding  cases where $r(g_{\vec{x},b_1}, g_{\vec{y},b_2}) = 0$, to get the following:
		\begin{align}
			M(\{g_{\vec{x},b_1}\}) &= \cP[\vec{x}] \sum_{j=0}^{\lfloor \sqrt{n} \rfloor-1} \sum_{\substack{g_{\vec{y},b_2} \in \mathcal{X}:\\ J(\vec{x},\vec{y}) = j\\ b_2 = 1-b_1\\good(\vec{y})}} \frac{\cP[\vec{y}]}{\cP[Head(\vec{y},j)]}\,. \label{eq:M_analysis_1_1}
		\end{align}		
		There is a one-to-one correspondence between walks $\vec{y} \in \cW$ and functions $g_{\vec{y},b_2}$ with $b_2 = 1-b_1$.
		Therefore we may equivalently sum over $\vec{y} \in \cW$ in \cref{eq:M_analysis_1_1}:
		\begin{align}
			M(\{g_{\vec{x},b_1}\})
			&= \cP[\vec{x}] \sum_{j=0}^{\lfloor \sqrt{n} \rfloor-1} \sum_{\substack{\vec{y} \in \cW:\\ J(\vec{x},\vec{y}) = j\\good(\vec{y})}} \frac{\cP[\vec{y}]}{\cP[Head(\vec{y},j)]} \,. \label{eq:M_analysis_1_2}
		\end{align}	
		Since $\vec{y}$ is the concatenation of $Head(\vec{y},j)$ with $Tail(\vec{y},j)$, we have:
		\begin{align}  \label{eq:reminder_P_y_decomposition_tail_head}
			\cP[\vec{y}] = \cP[Head(\vec{y},j)] \cdot \cP[Tail(\vec{y},j)] \cdot \cP_{y_{j\cdot T},\; y_{j \cdot T + 1}}\,.
		\end{align}
		Substituting \cref{eq:reminder_P_y_decomposition_tail_head} in \cref{eq:M_analysis_1_2} gives
		\begin{align}
			M(\{g_{\vec{x},b_1}\})
			&= \cP[\vec{x}] \sum_{j=0}^{\lfloor \sqrt{n} \rfloor-1} \sum_{\substack{\vec{y} \in \cW:\\ J(\vec{x},\vec{y}) = j\\ good(\vec{y})}} \cP[Tail(\vec{y},j)] \cdot \cP_{y_{j\cdot T},\; y_{j \cdot T + 1}} \,. \label{eq:M_analysis_1}
		\end{align}
		If it weren't for the restrictions that $J(\vec{x},\vec{y}) = j$ (instead of $J(\vec{x},\vec{y}) \ge j$) and $good(y)$, then the inner sum would be over all possible walks from $x_{jT}$, and would thus be $1$. Because of those restrictions, we instead invoke \cref{lem:y_counting_for_MZ} to continue from \cref{eq:M_analysis_1} as follows:
		\begin{align}
			M(\{g_{\vec{x},b_1}\})
			&\ge \cP[\vec{x}] \cdot 2^{-4\sigma} \cdot \lfloor \sqrt{n} \rfloor\,.
		\end{align}
		
		If $\vec{x}$ is bad, then $M(\{g_{\vec{x},b_1}\}) = 0$, and so the function $g_{\vec{x},b_1}$ does not contribute to $M(\cZ)$.
		Therefore
		\begin{align}
			M(\cZ) \ge 2^{-4\sigma} \lfloor \sqrt{n} \rfloor \sum_{\substack{g_{\vec{x},b_1} \in \cZ\\ good(\vec{x})}} \cP[\vec{x}] \,. \label{eq:M_bound}
		\end{align}
		
		\paragraph{Upper bounding $q(\cZ)$.}
		Define 
		\begin{align} \label{def:q_tilde} \widetilde{q}(\cZ,v) =  \sum_{g_{\vec{x},b_1} \in \cZ} \sum_{g_{\vec{y},b_2} \in \cZ} r(g_{\vec{x},b_1},g_{\vec{y},b_2}) \mathbbm{1}_{\{g_{\vec{x},b_1}(v) \ne g_{\vec{y},b_2}(v)\}}\,.
		\end{align}
		By definition of $q(\cZ)$, we have  
		\begin{align} 
			q(\cZ) = 
			\max_{v \in V} \sum_{g_{\vec{x},b_1} \in \cZ} \sum_{g_{\vec{y},b_2} \in \cZ} r(g_{\vec{x},b_1},g_{\vec{y},b_2}) \mathbbm{1}_{\{g_{\vec{x},b_1}(v) \ne g_{\vec{y},b_2}(v)\}}
			= \max_{v \in V} \widetilde{q}(\cZ,v)\,.
		\end{align}
		
		We will bound $q(\cZ)$ by bounding $\widetilde{q}(\cZ,v)$ for an arbitrary $v$.
		Fix an arbitrary vertex $v$.
		We first partition  the inner sum according to $J(\vec{x},\vec{y})$:
		\begin{align}
			\widetilde{q}(\cZ,v)
			&= \sum_{g_{\vec{x},b_1} \in \cZ} \sum_{g_{\vec{y},b_2} \in \cZ} r(g_{\vec{x},b_1},g_{\vec{y},b_2}) \mathbbm{1}_{\{g_{\vec{x},b_1}(v) \ne g_{\vec{y},b_2}(v)\}} \explain{By definition of $\widetilde{q}(\cZ,v)$.}\\
			&= \sum_{g_{\vec{x},b_1} \in \cZ} \sum_{j=0}^{\lfloor \sqrt{n} \rfloor-1} \sum_{\substack{g_{\vec{y},b_2} \in \cZ:\\ J(\vec{x},\vec{y}) = j}} r(g_{\vec{x},b_1},g_{\vec{y},b_2}) \mathbbm{1}_{\{g_{\vec{x},b_1}(v) \ne g_{\vec{y},b_2}(v)\}} \,. \label{eq:q_analysis_1_1}
		\end{align}
		
		By \cref{lem:r_tilde}, we may continue from \cref{eq:q_analysis_1_1} and expand to get
		\begin{align}
			\widetilde{q}(\cZ,v) &\le 2 \sum_{g_{\vec{x},b_1} \in \cZ} \sum_{j=0}^{\lfloor \sqrt{n} \rfloor-1} \sum_{\substack{g_{\vec{y},b_2} \in \cZ:\\ v \in Tail(\vec{y}, j)\\ J(\vec{x},\vec{y}) = j}} r(g_{\vec{x},b_1},g_{\vec{y},b_2}) \\
			&= 2 \sum_{\substack{g_{\vec{x},b_1} \in \cZ\\ good(\vec{x})}} \cP[\vec{x}] \sum_{j=0}^{\lfloor \sqrt{n} \rfloor-1} \sum_{\substack{g_{\vec{y},b_2} \in \cZ:\\ v \in Tail(\vec{y}, j) \\ J(\vec{x},\vec{y}) = j\\ b_2 = 1-b_1\\ good(\vec{y})}}  \frac{\cP[\vec{y}]}{\cP[Head(\vec{y},j)]}
			\explain{Using the definition of $r$.} \\ 
			&\le 2 \sum_{\substack{g_{\vec{x},b_1} \in \cZ\\ good(\vec{x})}} \cP[\vec{x}] \sum_{j=0}^{\lfloor \sqrt{n} \rfloor-1} \sum_{\substack{g_{\vec{y},b_2} \in \cX:\\ v \in Tail(\vec{y}, j) \\ J(\vec{x},\vec{y}) = j\\ b_2 = 1-b_1\\ good(\vec{y})}}  \frac{\cP[\vec{y}]}{\cP[Head(\vec{y},j)]} \,. 
			\explain{Since $\cZ \subseteq \cX$.} \\ \label{eq:q_analysis_2_2}
		\end{align}
		
		Again, there is a one-to-one correspondence between walks $\vec{y} \in \cW$ and functions $g_{\vec{y},b_2} \in \cX$ with $b_2 = 1-b_1$, so we may equivalently sum over $\vec{y} \in \cW$ in \cref{eq:q_analysis_2_2}.
		Additionally, we expand the scope from $J(\vec{x},\vec{y}) = j$ to $J(\vec{x},\vec{y}) \ge j$, which is equivalent to $Head(\vec{x},j) = Head(\vec{y},j)$.
		Finally, we drop the requirement that $good(\cY)$.
		Continuing from \cref{eq:q_analysis_2_2},	
		\begin{align}
			\widetilde{q}(\cZ,v) &\le 2 \sum_{\substack{g_{\vec{x},b_1} \in \cZ\\ good(\vec{x})}} \cP[\vec{x}] \sum_{j=0}^{\lfloor \sqrt{n} \rfloor-1} \sum_{\substack{\vec{y} \in \cW:\\ v \in Tail(\vec{y}, j) \\ J(\vec{x},\vec{y}) \ge j}}  \frac{\cP[\vec{y}]}{\cP[Head(\vec{y},j)]} \,. 
			\label{eq:q_analysis_2}
		\end{align}
		
		From here we partition based on where in $Tail(\vec{y},j)$ the vertex $v$ lies: the first short part of the tail, i.e. $Tail(\vec{y},j,j+1)$, or the rest of the tail, i.e. $Tail(\vec{y},j+1)$.
		Continuing from \cref{eq:q_analysis_2},
		\begin{align}
			\widetilde{q}(\cZ,v)
			&\le 2 \sum_{\substack{g_{\vec{x},b_1} \in \cZ\\ good(\vec{x})}} \cP[\vec{x}] \sum_{j=0}^{\lfloor \sqrt{n} \rfloor-1} \left( \sum_{\substack{\vec{y} \in \cW:\\ v \in Tail(\vec{y}, j, j+1) \\ J(\vec{x},\vec{y}) \ge j}}  \frac{\cP[\vec{y}]}{\cP[Head(\vec{y},j)]} + \sum_{\substack{\vec{y} \in \cW:\\ v \in Tail(\vec{y}, j+1)\\ J(\vec{x},\vec{y}) \ge j}} \frac{\cP[\vec{y}]}{\cP[Head(\vec{y},j)]} \right) \label{eq:qz_start_1}
		\end{align}
		
		Since $\vec{y}$ is the concatenation of $Head(\vec{y},j)$ with $Tail(\vec{y},j)$, we have 
		\begin{align} \label{eq:decompose_y_head_tail} \cP[\vec{y}] = \cP[Head(\vec{y},j)] \cdot \cP[Tail(\vec{y},j)] \cdot \cP_{y_{j\cdot T},\; y_{j \cdot T + 1}}\,.
		\end{align}
		Using \cref{eq:decompose_y_head_tail} in  \cref{eq:qz_start_1}, we obtain 
		\begin{small}
			\begin{equation}
				\begin{split}
					\widetilde{q}(\cZ,v) \le 2 \sum_{\substack{g_{\vec{x},b_1} \in \cZ\\ good(\vec{x})}} \cP[\vec{x}] \sum_{j=0}^{\lfloor \sqrt{n} \rfloor-1} &\Biggl( \sum_{\substack{\vec{y} \in \cW:\\ v \in Tail(\vec{y}, j, j+1) \\ J(\vec{x},\vec{y}) \ge j}}  \cP[Tail(\vec{y},j)] \cdot \cP_{y_{j\cdot T},\; y_{j \cdot T + 1}} + \sum_{\substack{\vec{y} \in \cW:\\ v \in Tail(\vec{y}, j+1)\\ J(\vec{x},\vec{y}) \ge j}} \cP[Tail(\vec{y},j)] \cdot \cP_{y_{j\cdot T},\; y_{j \cdot T + 1}} \Biggr) \label{eq:qZ_start}
				\end{split}
			\end{equation}
		\end{small}
		To bound the first part, we will use \cref{lem:first_tail_segment_unlikely_hit_v}.
		Since $\vec{x}$ is good, we have
		\begin{align}
			\sum_{j=0}^{\lfloor \sqrt{n} \rfloor-1} \sum_{\substack{\vec{y} \in \cW:\\ v \in Tail(\vec{y}, j, j+1) \\ J(\vec{x},\vec{y}) \ge j}} \cP[Tail(\vec{y},j)] \cdot \cP_{y_{j\cdot T},\; y_{j \cdot T + 1}}
			&= \sum_{j=0}^{\lfloor \sqrt{n} \rfloor-1} \mbox{P}_{visit}(x_{jT},v, T) \explain{By definition of $\mbox{P}_{visit}$.}\\
			&\le T \sigma \,. \explain{By \cref{lem:first_tail_segment_unlikely_hit_v} since $\vec{x}$ is good.}\\ \label{eq:first_segment}
		\end{align}
		
		To bound the second part, we first partition according to the possible values of $ \vec{y}_{(j+1)T}$.
		We have
		\begin{align}
			\sum_{j=0}^{\lfloor \sqrt{n} \rfloor-1} \sum_{\substack{\vec{y} \in \cW:\\ v \in Tail(\vec{y}, j+1) \\ J(\vec{x},\vec{y}) \ge j}}  \cP[Tail(\vec{y},j)] \cdot \cP_{y_{j\cdot T},\; y_{j \cdot T + 1}}
			&= \sum_{j=0}^{\lfloor \sqrt{n} \rfloor-1} \sum_{u \in V} \mbox{P}_{end}(x_{jT}, u, T) \mbox{P}_{visit}(u, v, L-(j+1)T) \label{eq:q_second_part_analysis_1}
		\end{align}
		The visiting probability $\mbox{P}_{visit}(u,v,\ell)$ is  increasing in $\ell$, so continuing from   \cref{eq:q_second_part_analysis_1} we get
		\begin{align}
			& \sum_{j=0}^{\lfloor \sqrt{n} \rfloor-1} \sum_{\substack{\vec{y} \in \cW:\\ v \in Tail(\vec{y}, j+1) \\ J(\vec{x},\vec{y}) \ge j}}  \cP[Tail(\vec{y},j)] \cdot \cP_{y_{j\cdot T},\; y_{j \cdot T + 1}} \notag \\
			& \qquad \le \sum_{j=0}^{\lfloor \sqrt{n} \rfloor-1} \sum_{u \in V} \mbox{P}_{end}(x_{jT}, u, T) \mbox{P}_{visit}(u, v, L)  \\
			& \qquad \le \sum_{j=0}^{\lfloor \sqrt{n} \rfloor-1} \sum_{u \in V} \Bigl( \pi(u) + | \mbox{P}_{end}(x_{jT}, u, T) - \pi(u)|\Bigr) 
			\mbox{P}_{visit}(u, v, L), \label{eq:q_second_part_analysis_2}
		\end{align}
		where in \cref{eq:q_second_part_analysis_2} we used the inequality $a \leq b + |a-b|$ for $a,b \geq 0$. 
		
		A random walk with starting vertex drawn from $\pi$ has probability $\pi(v)$ of being at $v$ at each step.
		Formally, for all $\ell \in \mathbbm{N}$ we have
		\begin{align} \label{eq:uniform_start}
			\sum_{u \in V} \pi(u) \mbox{P}_{end}(u,v,\ell) = \pi(v)\,.
		\end{align}
		Additionally, by union bound we have
		\begin{align} \label{eq:union_bound}
			\mbox{P}_{visit}(u,v,L) \le \sum_{\ell = 1}^L \mbox{P}_{end}(u,v,\ell)\,.
		\end{align}
		We have 
		\begin{align}
			&  \sum_{j=0}^{\lfloor \sqrt{n} \rfloor-1} \sum_{u \in V} \left( \pi(u) + | \mbox{P}_{end}(x_{jT}, u, T) - \pi(u)|\right) 
			\mbox{P}_{visit}(u, v, L)  \\
			& \qquad  \le \sum_{j=0}^{\lfloor \sqrt{n} \rfloor-1} \left( \sum_{\ell = 1}^L \sum_{u \in V} \pi(u) \mbox{P}_{end}(u,v,\ell) + \sum_{u \in V} | \mbox{P}_{end}(x_{jT}, u, T) - \pi(u)| \cdot  \mbox{P}_{visit}(u,v,L) \right) \explain{By \cref{eq:union_bound}.} \\
			& \qquad = \sum_{j=0}^{\lfloor \sqrt{n} \rfloor-1} \left( L \pi(v) + \sum_{u \in V} |  \mbox{P}_{end}(x_{jT}, u, T) - \pi(u)| \cdot \mbox{P}_{visit}(u,v,L) \right)\,. \explain{By \cref{eq:uniform_start}} \\  \label{eq:474673373}
		\end{align} 
		
		Since $T = t_{mix}(\sigma/(2n))$, we have
		\begin{align}
			\sum_{u \in V} |  \mbox{P}_{end}(x_{jT}, u, T) - \pi(u)| \cdot \mbox{P}_{visit} (u,v,L)
			&\le \max_{u \in V} \mbox{P}_{visit} (u,v,L) \sum_{u \in V} |  \mbox{P}_{end}(x_{jT}, u, T) - \pi(u)| \notag \\
			&= \max_{u \in V} \mbox{P}_{visit} (u,v,L) \sum_{u \in V} |  (\cP^T)_{x_{jT},u} - \pi(u)|\notag \\
			&\leq  \max_{u \in V} \mbox{P}_{visit}(u,v,L) \cdot \frac{\sigma}{n} \,. \label{eq:1343333}
		\end{align}
		Combining  \cref{eq:474673373} and \cref{eq:1343333} yields 
		\begin{align}
			& \sum_{j=0}^{\lfloor \sqrt{n} \rfloor-1} \sum_{u \in V} \left( \pi(u) + | \mbox{P}_{end}(x_{jT}, u, T) - \pi(u)|\right) 
			\mbox{P}_{visit}(u, v, L)  \le \sum_{j=0}^{\lfloor \sqrt{n} \rfloor-1} \left( L \pi(v) + \frac{\sigma}{n} \max_{u \in V} \mbox{P}_{visit}(u,v,L) \right) \\
			&\qquad \le \lfloor \sqrt{n} \rfloor (L+1) \frac{\sigma}{n} \explain{Since $\mbox{P}_{visit}(u,v,L) \le 1$ and $\pi(v) \le \frac{\sigma}{n}$.}\\
			&\qquad  \le 2T \sigma \,.  \explain{Since $L = T \lfloor \sqrt{n} \rfloor$ and $1 \le L$.}\\ \label{eq:second_segment}
		\end{align}
		Combining \cref{eq:q_second_part_analysis_2} and \cref{eq:second_segment}, we obtain 
		\begin{align}
			\sum_{j=0}^{\lfloor \sqrt{n} \rfloor-1} \sum_{\substack{\vec{y} \in \cW:\\ v \in Tail(\vec{y}, j+1) \\ J(\vec{x},\vec{y}) \ge j}}  \cP[Tail(\vec{y},j)] \cdot \cP_{y_{j\cdot T},\; y_{j \cdot T + 1}}  \le 2T \sigma \,. \label{eq:83838383}
		\end{align}
		
		Combining \cref{eq:qZ_start}, \cref{eq:first_segment}, and \cref{eq:83838383}, 
		\begin{align}
			\widetilde{q}(\cZ,v)
			&\le 2 \sum_{\substack{g_{\vec{x},b_1} \in \cZ\\ good(\vec{x})}} \cP[\vec{x}] 3T\sigma
			= 6 T \sigma \sum_{\substack{g_{\vec{x},b_1} \in \cZ\\ good(\vec{x})}} \cP[\vec{x}]\,.  
			\label{eq:rho_bound}
		\end{align}
		Since $q(\cZ) = \min_{v \in V} \widetilde{q}(\cZ,v)$ and \cref{eq:rho_bound} holds for an arbitrary choice of $v$, we get 
		\begin{align}
			q(\cZ) \le 6T \sigma \sum_{\substack{g_{\vec{x},b_1} \in \cZ\\ good(\vec{x})}} \cP[\vec{x}] \,. \label{eq:q_bound}
		\end{align}
		
		\paragraph{Bounding ${M(\cZ)}/{q(\cZ)}$.}
		
		Combining \cref{eq:M_bound} and \cref{eq:q_bound}, we obtain 
		\begin{align} \label{eq:proto_bound}
			\frac{M(\cZ)}{q(\cZ)}
			\ge \frac{ 2^{-4\sigma} \lfloor \sqrt{n} \rfloor \sum_{\substack{g_{\vec{x},b_1} \in \cZ\\ good(\vec{x})}} \cP[\vec{x}]}{6 T \sigma \sum_{\substack{g_{\vec{x},b_1} \in \cZ\\ good(\vec{x})}} \cP[\vec{x}]}
			= \frac{2^{-4\sigma}}{6\sigma} \cdot \frac{\lfloor \sqrt{n} \rfloor}{T}\,.
		\end{align}
		We have $2^{-4\sigma}/(6\sigma) \in \Omega(1/exp(3\sigma))$.
		Thus applying  \cref{thm:relational_adversary} to the expression in \cref{eq:proto_bound}, we obtain the required lower bound for local search on $G$, namely \[
		\Omega\left( \frac{\sqrt{n}}{t_{mix}\left(\frac{\sigma}{2n}\right)    \cdot \exp(3\sigma)}\right)\,.
		\]
		This completes the proof.
	\end{proof}
	
	\section*{Acknowledgements}
	We would like to thank Davin Choo for several useful discussions.
	
	\bibliographystyle{alpha}
	\bibliography{bibliography}

\newcommand{\etalchar}[1]{$^{#1}$}
\begin{thebibliography}{BDPX09}

\bibitem[Aar06]{Aaronson06}
Scott Aaronson.
\newblock Lower bounds for local search by quantum arguments.
\newblock {\em {SIAM} J. Comput.}, 35(4):804--824, 2006.

\bibitem[Ald83]{aldous1983minimization}
David Aldous.
\newblock Minimization algorithms and random walk on the $ d $-cube.
\newblock {\em The Annals of Probability}, 11(2):403--413, 1983.

\bibitem[Ama13]{amari_book}
Shun-ichi Amari.
\newblock Information geometry and its applications: Survey.
\newblock In Frank Nielsen and Fr{\'e}d{\'e}ric Barbaresco, editors, {\em
  Geometric Science of Information}, pages 3--3, Berlin, Heidelberg, 2013.
  Springer Berlin Heidelberg.

\bibitem[BCR24]{BCR23}
Simina Br{\^a}nzei, Davin Choo, and Nicholas Recker.
\newblock The sharp power law of local search on expanders.
\newblock In {\em Proceedings of the ACM-SIAM Symposium on Discrete Algorithms
  (SODA)}, 2024.
\newblock https://arxiv.org/abs/2305.08269.

\bibitem[BDN19]{babichenko2019communication}
Yakov Babichenko, Shahar Dobzinski, and Noam Nisan.
\newblock The communication complexity of local search.
\newblock In {\em Proceedings of the 51st Annual ACM SIGACT Symposium on Theory
  of Computing}, pages 650--661, 2019.

\bibitem[BDPX05]{boyd2005symmetry}
Stephen Boyd, Persi Diaconis, Pablo Parrilo, and Lin Xiao.
\newblock Symmetry analysis of reversible markov chains.
\newblock {\em Internet Mathematics}, 2(1):31--71, 2005.

\bibitem[BDPX09]{BDPX09}
S.~Boyd, P.~Diaconis, P.~Parrilo, and L.~Xiao.
\newblock Fastest mixing markov chain on graphs with symmetries.
\newblock {\em SIAM J. Optimization}, 20:792--819, 2009.

\bibitem[BDX04]{boyd2004fastest}
Stephen Boyd, Persi Diaconis, and Lin Xiao.
\newblock Fastest mixing markov chain on a graph.
\newblock {\em SIAM review}, 46(4):667--689, 2004.

\bibitem[BL22]{BL22}
Simina Br{\^a}nzei and Jiawei Li.
\newblock The query complexity of local search and brouwer in rounds.
\newblock In {\em Conference on Learning Theory}, pages 5128--5145. PMLR, 2022.

\bibitem[BM20]{bubeck2020trap}
S{\'e}bastien Bubeck and Dan Mikulincer.
\newblock How to trap a gradient flow.
\newblock In {\em Conference on Learning Theory}, pages 940--960. PMLR, 2020.

\bibitem[Bon13]{bonnabel2013stochastic}
Silvere Bonnabel.
\newblock Stochastic gradient descent on riemannian manifolds.
\newblock {\em IEEE Transactions on Automatic Control}, 58(9):2217--2229, 2013.

\bibitem[CD05]{chen2005algorithms}
Xi~Chen and Xiaotie Deng.
\newblock On algorithms for discrete and approximate brouwer fixed points.
\newblock In {\em Proceedings of the thirty-seventh annual ACM symposium on
  Theory of computing}, pages 323--330, 2005.

\bibitem[CD09]{CD09}
Xi~Chen and Xiaotie Deng.
\newblock On the complexity of 2d discrete fixed point problem.
\newblock {\em Theor. Comput. Sci.}, 410(44):4448--4456, 2009.

\bibitem[CDHS20]{stationary_I}
Yair Carmon, John~C. Duchi, Oliver Hinder, and Aaron Sidford.
\newblock Lower bounds for finding stationary points {I}.
\newblock {\em Math. Program.}, 184(1):71--120, 2020.

\bibitem[CDHS21]{stationary_II}
Yair Carmon, John~C. Duchi, Oliver Hinder, and Aaron Sidford.
\newblock Lower bounds for finding stationary points {II:} first-order methods.
\newblock {\em Math. Program.}, 185(1-2):315--355, 2021.

\bibitem[CDT09]{chen2009settling}
Xi~Chen, Xiaotie Deng, and Shang-Hua Teng.
\newblock Settling the complexity of computing two-player nash equilibria.
\newblock {\em Journal of the ACM (JACM)}, 56(3):1--57, 2009.

\bibitem[Chu16]{chuzhoy2016routing}
Julia Chuzhoy.
\newblock Routing in undirected graphs with constant congestion.
\newblock {\em SIAM Journal on Computing}, 45(4):1490--1532, 2016.

\bibitem[CPY17]{chen2017complexity}
Xi~Chen, Dimitris Paparas, and Mihalis Yannakakis.
\newblock The complexity of non-monotone markets.
\newblock {\em Journal of the ACM (JACM)}, 64(3):1--56, 2017.

\bibitem[CT07]{chen2007paths}
Xi~Chen and Shang-Hua Teng.
\newblock Paths beyond local search: A tight bound for randomized fixed-point
  computation.
\newblock In {\em 48th Annual IEEE Symposium on Foundations of Computer Science
  (FOCS'07)}, pages 124--134. IEEE, 2007.

\bibitem[DGP09]{daskalakis2009complexity}
Constantinos Daskalakis, Paul~W Goldberg, and Christos~H Papadimitriou.
\newblock The complexity of computing a nash equilibrium.
\newblock {\em Communications of the ACM}, 52(2):89--97, 2009.

\bibitem[DP11]{daskalakis2011continuous}
Constantinos Daskalakis and Christos Papadimitriou.
\newblock Continuous local search.
\newblock In {\em Proceedings of the twenty-second annual ACM-SIAM symposium on
  Discrete algorithms}, pages 790--804. SIAM, 2011.

\bibitem[DR10]{dinh2010quantum}
Hang Dinh and Alexander Russell.
\newblock Quantum and randomized lower bounds for local search on
  vertex-transitive graphs.
\newblock {\em Quantum Information \& Computation}, 10(7):636--652, 2010.

\bibitem[DS20]{pmlr-v119-drori20a}
Yoel Drori and Ohad Shamir.
\newblock The complexity of finding stationary points with stochastic gradient
  descent.
\newblock In Hal~Daumé III and Aarti Singh, editors, {\em Proceedings of the
  37th International Conference on Machine Learning}, volume 119 of {\em
  Proceedings of Machine Learning Research}, pages 2658--2667. PMLR, 13--18 Jul
  2020.

\bibitem[DSZ21]{costis_minmax}
Constantinos Daskalakis, Stratis Skoulakis, and Manolis Zampetakis.
\newblock The complexity of constrained min-max optimization.
\newblock In {\em Proceedings of the 53rd Annual ACM SIGACT Symposium on Theory
  of Computing}, page 1466–1478, 2021.

\bibitem[FGHS22]{fearnley2022complexity}
John Fearnley, Paul Goldberg, Alexandros Hollender, and Rahul Savani.
\newblock {The complexity of gradient descent: $CLS = PPAD \cap PLS$}.
\newblock {\em Journal of the ACM}, 70(1):1--74, 2022.

\bibitem[HPV89]{hirsch1989exponential}
Michael~D Hirsch, Christos~H Papadimitriou, and Stephen~A Vavasis.
\newblock Exponential lower bounds for finding brouwer fix points.
\newblock {\em Journal of Complexity}, 5(4):379--416, 1989.

\bibitem[JPY88]{DBLP:journals/jcss/JohnsonPY88}
David~S. Johnson, Christos~H. Papadimitriou, and Mihalis Yannakakis.
\newblock How easy is local search?
\newblock {\em J. Comput. Syst. Sci.}, 37(1):79--100, 1988.

\bibitem[LP17]{levin2017markov}
David~A Levin and Yuval Peres.
\newblock {\em Markov chains and mixing times}.
\newblock American Mathematical Society; 2nd Revised edition, 2017.

\bibitem[LTT89]{llewellyn1989local}
Donna~Crystal Llewellyn, Craig Tovey, and Michael Trick.
\newblock Local optimization on graphs.
\newblock {\em Discrete Applied Mathematics}, 23(2):157--178, 1989.

\bibitem[OTZ21]{olesker2021geometric}
Sam Olesker-Taylor and Luca Zanetti.
\newblock Geometric bounds on the fastest mixing markov chain.
\newblock {\em arXiv preprint arXiv:2111.05816}, 2021.

\bibitem[Pap94]{Papadimitriou_1994}
Christos~H. Papadimitriou.
\newblock On the complexity of the parity argument and other inefficient proofs
  of existence.
\newblock {\em J. Comput. Syst. Sci.}, 48(3):498--532, 1994.

\bibitem[SS04]{santha2004quantum}
Miklos Santha and Mario Szegedy.
\newblock Quantum and classical query complexities of local search are
  polynomially related.
\newblock In {\em Proceedings of the thirty-sixth annual ACM symposium on
  Theory of computing}, pages 494--501, 2004.

\bibitem[SY09]{sun2009quantum}
Xiaoming Sun and Andrew Chi-Chih Yao.
\newblock On the quantum query complexity of local search in two and three
  dimensions.
\newblock {\em Algorithmica}, 55(3):576--600, 2009.

\bibitem[Tov81]{tovey81}
Craig Tovey.
\newblock Polynomial local improvement algorithms in combinatorial
  optimization, 1981.
\newblock Ph.D. thesis, Stanford University.

\bibitem[Vav93]{vavasis1993black}
Stephen~A Vavasis.
\newblock Black-box complexity of local minimization.
\newblock {\em SIAM Journal on Optimization}, 3(1):60--80, 1993.

\bibitem[Ver06]{Verhoeven06}
Yves~F. Verhoeven.
\newblock Enhanced algorithms for local search.
\newblock {\em Information Processing Letters}, 97(5):171--176, 2006.

\bibitem[VY11]{vazirani2011market}
Vijay~V Vazirani and Mihalis Yannakakis.
\newblock Market equilibrium under separable, piecewise-linear, concave
  utilities.
\newblock {\em Journal of the ACM (JACM)}, 58(3):1--25, 2011.

\bibitem[Zha09]{zhang2009tight}
Shengyu Zhang.
\newblock Tight bounds for randomized and quantum local search.
\newblock {\em SIAM Journal on Computing}, 39(3):948--977, 2009.

\bibitem[ZLJ{\etalchar{+}}20]{pmlr-v119-zhang20p}
Jingzhao Zhang, Hongzhou Lin, Stefanie Jegelka, Suvrit Sra, and Ali Jadbabaie.
\newblock Complexity of finding stationary points of nonconvex nonsmooth
  functions.
\newblock In Hal~Daumé III and Aarti Singh, editors, {\em Proceedings of the
  37th International Conference on Machine Learning}, volume 119 of {\em
  Proceedings of Machine Learning Research}, pages 11173--11182. PMLR, 13--18
  Jul 2020.

\end{thebibliography}
	
	\newpage 
	
	\appendix
	\section{Helper lemmas} \label{sec:helper_lemmas}

	\begin{mylemma} \label{lem:proof_f_X_is_valid}
		For each walk  $\X = (x_0, \ldots, x_{L}) \in \cW$, the function $f_{\X}$ given by  Definition~\ref{def:f_x_and_g_x_b} is valid for the walk ${\X}$.
	\end{mylemma}
	\begin{proof}
		We show  the  conditions required by the definition of a valid function (\cref{def:valid_function}) hold.
		
		\begin{description}
			\item[Condition 1.] 
			Consider two arbitrary vertices $v_1,v_2 \in {\X}$. Define 
			\[
			i_1 = \max\bigl\{k \in \{0, \ldots, L\} \mid v_1 = x_k \bigr\}\qquad \mbox{and} \qquad i_2 = \max\bigl\{k \in \{0, \ldots, L\} \mid v_2 = x_k\bigr\} \,.
			\]
			Without loss of generality, we have $i_1 \leq i_2$. Condition $1$ requires that if $i_1 < i_2$, then $f_{\vec{x}}(v_1) > f_{\vec{x}}(v_2)$. 
			Since $v_1, v_2 \in \vec{x}$, \cref{def:f_x_and_g_x_b} of  the  function $f_{\vec{x}}$ states that  $f_{\vec{x}}(v_1) = -i_1$ and $f_{\vec{x}}(v_2) = -i_2$. Thus if $i_1 < i_2$, then $f_{\vec{x}}(v_1) > f_{\vec{x}}(v_2)$, as required. 
			
			\item[Condition 2.] The second condition requires that  $f_{\vec{x}}(v) = dist(x_0, v) > 0$ for all $v \in V \setminus \vec{x}$.
			Since $\vec{x} \in \cW$, we have $x_0 = 1$. Using  \cref{def:f_x_and_g_x_b} we have $f_{\vec{x}}(v) = dist(v,1) = dist(v,x_0) > 0$ for all $v \not \in \vec{x}$.
			
			\item[Condition 3.] The third condition requires that $f_{\vec{x}}(x_i) \leq 0$ for all $i \in \{0, \ldots, L\}$. This condition is clearly met since
			$f_{\vec{x}}(v) \in \{0, -1, \ldots, -L\}$ for all $v \in \vec{x}$. 
		\end{description}
		Therefore $f_{\X}$ is valid for the walk ${\X}$.
	\end{proof}

	\begin{lemma} \label{lem:mixing_time_sigma}
		Let $(w_0, w_1, \ldots)$ be a Markov chain generated by transition matrix $\cP$ with arbitrary starting distribution.
		Then for all $v \in V$ we have
		\begin{align*}
			\Pr[w_{t_{mix}(\sigma/(2n))} = v] \le 2\sigma/n \,.
		\end{align*}
	\end{lemma}
	\begin{proof}
		By definition of $t_{mix}(\sigma/(2n))$ we have
		\begin{align} \label{eq:mixing_time_sigma_1}
			\sum_{v \in V} \left|\Pr[w_{t_{mix}(\sigma/(2n))} = v] - \pi(v)\right| \le \sigma/n\,.
		\end{align}
		
		Since each term of the sum in \cref{eq:mixing_time_sigma_1} is non-negative, we have
		
		\begin{align} \label{eq:mixing_time_sigma_2}
			\forall v \in V\; \left|\Pr[w_{t_{mix}(\sigma/(2n))} = v] - \pi(v)\right| \le \sigma/n \,.
		\end{align}
		
		Removing the absolute value from \cref{eq:mixing_time_sigma_2} and rearranging, we get
		
		\begin{align} \label{eq:mixing_time_sigma_3}
			\forall v \in V\; \Pr[w_{t_{mix}(\sigma/(2n))} = v] &\le \sigma/n + \pi(v) \notag \\
			&\le 2\sigma/n \,. \explain{Since $\sigma \ge \pi(v) \cdot n$.}
		\end{align}
		
		This concludes the proof of the lemma.
	\end{proof}
	
	\begin{lemma} \label{lem:first_tail_segment_unlikely_hit_v}
		Let $S \subseteq V$ be a subset of vertices.
		Let $v \in V$ be a vertex and 
		$\ell \in \mathbbm{N}$.
		Then
		\begin{align*}
			\sum_{u \in S} \mbox{\em P}_{visit}(u,v,\ell) \le \ell\sigma \,.
		\end{align*}
	\end{lemma}
	\begin{proof}
		
		Let $T_\ell$ be the random variable representing the number of times a random walk of length $\ell$ starting at $v$ visits a vertex in $S$.
		Decomposing by vertices in $S$, we have
		\begin{align}
			\sum_{u \in S} \mbox{P}_{visit}(u,v,\ell) &\le \sum_{u \in S} \mbox{E}_{visit}(u,v,\ell) \\
			&= \sum_{u \in S} \mbox{E}_{visit}(v,u,\ell) \frac{\pi(v)}{\pi(u)}\,. \explain{By \cref{lem:reversibility_markov_chain}}\\
			\label{eq:bounding_sum_of_p_visit_by_expected_visit_times_ratio_fractions}
		\end{align}
		Using the definition of $\sigma = \max_{u,v \in V} {\pi(v)}/{ \pi(u)}$, we have 
		\begin{align}
			\sum_{u \in S} \mbox{E}_{visit}(v,u,\ell) \frac{\pi(v)}{\pi(u)} &\le \sigma\sum_{u \in S} \mbox{E}_{visit}(v,u,\ell) \notag \\
			&= {\sigma  \cdot \Ex[T_\ell]} \explain{By definition of $T_\ell$.}\\
			&\le \ell \sigma\,.  \explain{Since $T_\ell \le \ell$.} \\
			\label{eq:bounding_sum_of_expected_visits_over_S}
		\end{align}
		Combining \cref{eq:bounding_sum_of_p_visit_by_expected_visit_times_ratio_fractions} and \cref{eq:bounding_sum_of_expected_visits_over_S}, we get $\sum_{u \in S} \mbox{P}_{visit}(u,v,\ell) \leq \ell \sigma$, as required.
	\end{proof}
	
	\begin{lemma} \label{lem:y_counting_for_MZ}
		
		Let $n \ge 16 \sigma^2$.
		Fix a good walk $\vec{x} = (x_0, \ldots, x_L)$ with $x_0 = 1$ and $\cP_{x_i, x_{i+1}} > 0$ for all $0 \le i < L$.
		Then for each $0 \le j < \lfloor \sqrt{n} \rfloor$, 
		\begin{align}
			\sum_{\substack{\vec{y} \in \cW:\\ J(\vec{x},\vec{y}) = j\\ good(\vec{y})}} \cP[Tail(\vec{y},j)] \cdot \cP_{y_{j\cdot T},\; y_{j \cdot T + 1}} \ge 2^{-4\sigma} \,. \label{eq:target_inequality_collision_J_x_y_j}
		\end{align}
	\end{lemma}
	\begin{proof}
		
		Let $P_\cW$ be the distribution over the set of walks  $\cW$ generated by sampling a walk according to $\cP$ starting at the vertex $1$ with $L= \lfloor \sqrt{n} \rfloor \cdot T$ edges. Let $\vec{z}$ be a random walk drawn from $P_\cW$. Recall that every $T$-th vertex of a walk is called a milestone, and a walk is good if it does not repeat milestones. 
		We have
		\begin{align} \label{eq:y_counting_for_MZ_0}
			\sum_{\substack{\vec{y} \in \cW:\\ J(\vec{x},\vec{y}) = j\\ good(\vec{y})}} \cP[Tail(\vec{y},j)] \cdot \cP_{y_{j\cdot T},\; y_{j \cdot T + 1}}   & = \sum_{\substack{\vec{y} \in \cW:\\ J(\vec{x},\vec{y}) = j\\ good(\vec{y})}}  \cP_{y_{j\cdot T},\; y_{j \cdot T + 1}} \cdot \cP_{y_{j \cdot T+1}, y_{j \cdot T+2}} \cdot \ldots \cdot \cP_{y_{L-1}, y_{L}} \\ 
			%
			&= \Pr\Bigl[good(\vec{z}) \wedge J(\vec{x},\vec{z}) = j \;|\; Head(\vec{x}, j) = Head(\vec{z}, j)\Bigr]\,.
		\end{align}
		
		\begin{figure}[h!]
			\centering
			\begin{minipage}[t]{0.49\textwidth}
				\includegraphics[width=\textwidth]{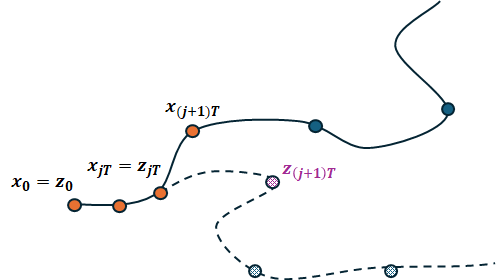}
				\caption{The milestone $z_{(j+1)T}$, shown in purple, may not match any of the orange milestones: $x_0$ through $x_{jT}$ because it would make $\vec{z}$ bad, and $x_{(j+1)T}$ because it would make $J(\vec{x},\vec{z}) > j$.}
			\end{minipage}
			\hfill
			\begin{minipage}[t]{0.49\textwidth}
				\includegraphics[width=\textwidth]{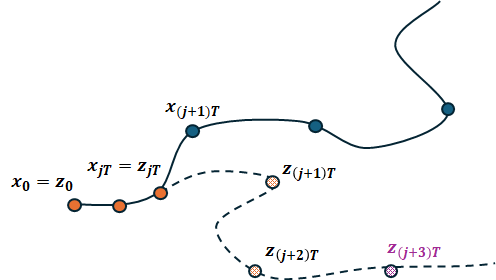}
				\caption{The milestone $z_{(j+3)T}$, shown in purple, may not match any of the orange milestones because it would make $\vec{z}$ bad.}
			\end{minipage}
		\end{figure}

		Equivalently, we can sample $\vec{z} \sim P_\cW$ one segment at a time with the constraint that the initial $jT+1$ vertices of $\vec{z}$ must match those of $\vec{x}$. That is, set $Head(\vec{z}, j) = Head(\vec{x}, j)$. Then for \\ $i = j, \ldots,\lfloor \sqrt{n} \rfloor$, sample the segment $Tail(\vec{z},i-1,i)$ conditioned on having set $Head(\vec{z},i-1)$. 

		For every $0 \le k \le \lfloor \sqrt{n} \rfloor$, the set of vertices given by the first $k+1$ milestones  of $\vec{z}$ is 
		\begin{align} 
			S_k = \{z_0, z_T, \ldots, z_{kT}\}\,.
		\end{align}
		
        Given that $Head(\vec{x},j) = Head(\vec{z},j)$, a sufficient condition for $J(\vec{x},\vec{z}) = j$ is
		\begin{align} \label{eq:z_j_plus_one_not_same_as_x_j_plus_one}
			z_{(j+1)T} \neq x_{(j+1)T}\,.
		\end{align}
		Similarly, given that $Head(\vec{x},j) = Head(\vec{z},j)$, the condition that $\vec{z}$ is good is equivalent to 
		\begin{align} \label{eq:z_kt_not_in_s_k_minus_one}
			z_{kT} \not \in S_{k-1} \qquad \forall k \in \{j+1, j+2, \ldots, \lfloor \sqrt{n} \rfloor\} \,.
		\end{align}
		
		
		For each $i \in \{j+1, \ldots, \lfloor \sqrt{n} \rfloor$, define 
		\begin{align} \label{def:Q_i}
			Q_i = 
			\begin{cases}
				\{x_{(j+1)T}\} \cup S_{j} & \text{if} \; \; i = j+1 \\ 
				S_{i-1} & \text{if} \; \;  j+1 < i \le \lfloor \sqrt{n} \rfloor \,.
			\end{cases}
		\end{align}
		Combining \eqref{eq:z_j_plus_one_not_same_as_x_j_plus_one} and \eqref{eq:z_kt_not_in_s_k_minus_one}, we get 
		\begin{align} \label{eq:y_counting_for_MZ_1}
			&\Pr\Bigl[good(\vec{z}) \wedge J(\vec{x},\vec{z}) = j \;|\; Head(\vec{x}, j) = Head(\vec{z}, j)\Bigr] \notag \\
			& \qquad      \ge \Pr\Bigl[ \Bigl( z_{kT} \not \in S_{k-1} \; \forall k  \in \{j+1, \ldots,  \lfloor \sqrt{n} \rfloor   \} \Bigr) \land  \Bigl( z_{(j+1)T} \neq x_{(j+1)T}\ \Bigr)  \mid  Head(\vec{x}, j) = Head(\vec{z}, j) \Bigr] \\
			& \qquad = \Pr\left[\bigwedge_{i=j+1}^{\lfloor \sqrt{n} \rfloor} z_{iT} \notin Q_{i} \mid Head(\vec{x},j) = Head(\vec{z},j)\right] \explain{By definition of $Q_i$ in \eqref{def:Q_i}}\\
			& \qquad = \prod_{i=j+1}^{\lfloor \sqrt{n} \rfloor} \Pr\left[z_{iT} \notin Q_{i} \mid \left(Head(\vec{x},j) = Head(\vec{z},j)\right) \wedge \bigl(\bigwedge_{k=j+1}^{i-1} z_{kT} \notin Q_{k} \bigr)\right]
		\end{align}
		For all $i \le \lfloor \sqrt{n} \rfloor$, let $\cW_i$ be the space of walks of length $Ti$ that can occur with positive probability under transition matrix $\cP$, formally defined as:
		\begin{align}
			\cW_i = \Bigl\{ \vec{w} \mid \vec{w} = (w_0,  \ldots, w_{T\cdot i}) \text{ where }w_0 = 1 \text{ and } \cP_{w_k, w_{k+1}} > 0 \text{ for all } 0 \le k < T\cdot i \Bigr\}\,.
		\end{align}
		Then, using \cref{eq:y_counting_for_MZ_1} gives 
		\begin{align} \label{eq:y_counting_for_MZ_1_2}
			&\Pr\Bigl[good(\vec{z}) \wedge J(\vec{x},\vec{z}) = j \;|\; Head(\vec{x}, j) = Head(\vec{z}, j)\Bigr]  \notag \\
			& \qquad \ge \prod_{i=j+1}^{\lfloor \sqrt{n} \rfloor} \min_{\vec{\xi} \in \cW_{i-1}} \Pr\Bigl[z_{iT} \notin Q_{i} \mid Head(\vec{z},i-1) = \vec{\xi}\Bigr] \,.
		\end{align}
		
		Since $T = t_{mix}(\sigma/(2n))$, \cref{lem:mixing_time_sigma} tells us that
		for all $v \in V$ and $1 \le i \le \lfloor \sqrt{n} \rfloor$
		\begin{align}
			\min_{\xi \in \cW_{i-1}} \Pr\Bigl[z_{iT} = v \mid Head(\vec{z}, i-1) = \xi \Bigr] \le \frac{2\sigma}{n}\,. \label{eq:zit_1}
		\end{align}
		By the union bound applied to \cref{eq:zit_1}, for each set $R \subseteq V$  and $1 \le i \le \lfloor \sqrt{n} \rfloor$, we have
		\begin{align}
			\min_{\xi \in \cW_{i-1}} \Pr\Bigl[z_{iT} \notin R \mid Head(\vec{z}, i-1) = \xi \Bigr] \ge 1 - \frac{2\sigma|R|}{n} \,. \label{eq:zit}
		\end{align}
		
		We consider two cases based on  $j$:
		
		\paragraph{Case 1: $j < \lfloor \sqrt{n} \rfloor-1$.}

		Since $j < \lfloor \sqrt{n} \rfloor-1$, we have 
		\begin{align}  \label{eq:x_j_plus_one_T_union_S_j}
			|\{x_{(j+1)T}\} \cup S_j| \le \lfloor \sqrt{n} \rfloor\,.
		\end{align}
		Furthermore, for all $j < k < \lfloor \sqrt{n} \rfloor$ we have 
		\begin{align}\label{eq:ub_S_k_for_k_more_than_j}
			|S_k| \le k+1 \le \lfloor \sqrt{n} \rfloor\,.
		\end{align}
		Combining \cref{eq:x_j_plus_one_T_union_S_j} and \cref{eq:ub_S_k_for_k_more_than_j} gives
		\begin{align}
			|Q_i| \le \lfloor \sqrt{n} \rfloor \qquad \forall j+1 \le i \le \lfloor \sqrt{n} \rfloor\,. \label{eq:Qi_small_case_1}
		\end{align} 
		We obtain: 
		\begin{align}
			& \Pr\Bigl[good(\vec{z}) \wedge J(\vec{x},\vec{z}) = j \;|\; Head(\vec{x}, j) = Head(\vec{z}, j)\Bigr] \notag \\
			& \qquad \geq \prod_{i=j+1}^{\lfloor \sqrt{n} \rfloor} \min_{\vec{\xi} \in \cW_{i-1}} \Pr\Bigl[z_{iT} \notin Q_{i} \mid Head(\vec{z},i-1) = \vec{\xi}\Bigr] \explain{By \cref{eq:y_counting_for_MZ_1_2}} \\ 
			& \qquad \ge \left( 1-\frac{2\sigma\lfloor \sqrt{n} \rfloor}{n}\right)^{\lfloor \sqrt{n} \rfloor-j} \explain{By  \cref{eq:zit} with $R = Q_i$ and since   $|Q_i| \leq \sqrt{n}$ by \cref{eq:Qi_small_case_1}.}\\
			& \qquad \ge \left(1-\frac{2\sigma}{\sqrt{n}}\right)^{\sqrt{n}} \,. \label{eq:y_counting_for_MZ_2_case_1_1}
		\end{align}
		
		By \cref{lem:derivative}, we have that $\left(1-\frac{2\sigma}{\sqrt{n}}\right)^{\sqrt{n}}$ is an increasing function of $\sqrt{n}$ since $n \geq 16 \sigma^2$.
		Therefore it is minimized at $\sqrt{n} = 4\sigma$.
		Substituting in  \cref{eq:y_counting_for_MZ_2_case_1_1}, we get 
		\begin{align}
			\Pr\Bigl[good(\vec{z}) \wedge J(\vec{x},\vec{z}) = j \;|\; Head(\vec{x}, j) = Head(\vec{z}, j)\Bigr]
			&\ge \left( 1 - \frac{2\sigma}{4\sigma}\right)^{4\sigma} = 2^{-4\sigma}\,. \label{eq:y_counting_for_MZ_2_case_1}
		\end{align}
		
		\paragraph{Case 2: $j = \lfloor \sqrt{n} \rfloor-1$.}
		We again invoke \cref{eq:zit}.
		In this case,
		\begin{align} \label{eq:Q_rtn}
			|Q_{\lfloor \sqrt{n} \rfloor}| \le |S_{\lfloor \sqrt{n} \rfloor - 1}| +1 \le \sqrt{n} + 1 \,.
		\end{align}
		
		Using \cref{eq:y_counting_for_MZ_1_2}, we obtain: 
		\begin{align}
			& \Pr\Bigl[good(\vec{z}) \wedge J(\vec{x},\vec{z}) = j \;|\; Head(\vec{x}, j) = Head(\vec{z}, j)\Bigr] \notag \\
			& \qquad \geq \min_{\vec{\xi} \in \cW_{\lfloor \sqrt{n} \rfloor -1}} \Pr\Bigl[z_{\lfloor \sqrt{n} \rfloor T} \notin Q_{\lfloor \sqrt{n} \rfloor} \mid Head(\vec{z},\lfloor \sqrt{n} \rfloor-1) = \vec{\xi}\Bigr]\,. \label{eq:result1_of_applying_y_counting_for_MZ_1_2}
		\end{align}

		Using \cref{eq:zit} with $R = Q_{\lfloor \sqrt{n} \rfloor}$ and since $|Q_{\lfloor \sqrt{n} \rfloor} \le \sqrt{n}+1$ by \cref{eq:Q_rtn}, we can lower bound the right hand side of  \cref{eq:result1_of_applying_y_counting_for_MZ_1_2} as follows:
		\begin{align}
			\min_{\vec{\xi} \in \cW_{\lfloor \sqrt{n} \rfloor -1}} \Pr\Bigl[z_{\lfloor \sqrt{n} \rfloor T} \notin Q_{\lfloor \sqrt{n} \rfloor} \mid Head(\vec{z},\lfloor \sqrt{n} \rfloor-1) = \vec{\xi}\Bigr]
			& \ge  1-\frac{2\sigma(\lfloor \sqrt{n} \rfloor+1)}{n} \,. \label{eq:result2_of_applying_y_counting_for_MZ_1_2}
		\end{align}
		
		Combining \cref{eq:result1_of_applying_y_counting_for_MZ_1_2} and \cref{eq:result2_of_applying_y_counting_for_MZ_1_2}, we get 
		\begin{align}
			\Pr\Bigl[good(\vec{z}) \wedge J(\vec{x},\vec{z}) = j \;|\; Head(\vec{x}, j) = Head(\vec{z}, j)\Bigr] \geq  1-\frac{2\sigma(\lfloor \sqrt{n} \rfloor+1)}{n}\,.  \label{eq:result3_of_applying_y_counting_for_MZ_1_2}
		\end{align}
		
		We can further lower bound the right hand side of \cref{eq:result3_of_applying_y_counting_for_MZ_1_2} as follows:
		\begin{align}
			1-\frac{2\sigma(\lfloor \sqrt{n} \rfloor+1)}{n} 
			& \ge 1 - \frac{\sqrt{n}(\sqrt{n}+1)}{2n} \explain{Since $\sqrt{n} \ge 4\sigma$.}\\
			& \ge \frac{1}{4} \explain{Since $n \ge 16\sigma^2 \ge 16$.}\\
			& \ge 2^{-4\sigma}\,. \explain{Since $\sigma \ge 1$.} \\ \label{eq:almost_there_y_counting_for_MZ_2_case_2}
		\end{align}
		
		Combining \cref{eq:result3_of_applying_y_counting_for_MZ_1_2} and \cref{eq:y_counting_for_MZ_2_case_2}, we get 
		\begin{align}
			\Pr\Bigl[good(\vec{z}) \wedge J(\vec{x},\vec{z}) = j \;|\; Head(\vec{x}, j) = Head(\vec{z}, j)\Bigr] \geq 2^{-4\sigma} \,. \label{eq:y_counting_for_MZ_2_case_2}
		\end{align}
		In both cases, we obtained  (by \cref{eq:y_counting_for_MZ_2_case_1} and  \cref{eq:y_counting_for_MZ_2_case_2})
		\begin{align} \label{eq:y_counting_for_MZ_2}
			\Pr\Bigl[good(\vec{z}) \wedge J(\vec{x},\vec{z}) = j \;|\; Head(\vec{x}, j) = Head(\vec{z}, j) \Bigr] \ge 2^{-4\sigma}\,.
		\end{align}
		
		Combining \cref{eq:y_counting_for_MZ_0} with \cref{eq:y_counting_for_MZ_2} yields
		\begin{align}
			\sum_{\substack{\vec{y} \in \cW:\\ J(\vec{x},\vec{y}) = j\\ good(\vec{y})}} \cP[Tail(\vec{y},j)] \cdot \cP_{y_{j\cdot T},\; y_{j \cdot T + 1}}
			\ge 2^{-4\sigma}\,. \notag
		\end{align}		
		This concludes the proof of the lemma.
	\end{proof}

	The next lemma is inspired by Lemma 9 from \cite{BCR23}.
	However, it is slightly different, so we include the proof here.
	\begin{lemma} \label{lem:r_tilde}
		Let $v \in V$ and $\cZ \subseteq \cX$.
		Then we have
		\[
		\sum_{g_{\vec{x},b_1}, g_{\vec{y},b_2} \in \mathcal{Z} } r(g_{\vec{x},b_1},g_{\vec{y},b_2}) 
		\mathbbm{1}_{\{g_{\vec{x},b_1}(v) \ne g_{\vec{y},b_2}(v)\}}
		\le 2 \sum_{\substack{g_{\vec{x},b_1},g_{\vec{y},b_2} \in \mathcal{Z}:\\ v \in Tail(\vec{y},J(\vec{x},\vec{y}))}} r(g_{\vec{x},b_1},g_{\vec{y},b_2}) \,. 
		\]
	\end{lemma}
	\begin{proof}
		
		If $g_{\vec{x},b_1}(v) \ne g_{\vec{y},b_2}(v)$, then either:
		\begin{itemize}
			\item $v \in Tail(\vec{x}, J(\vec{x},\vec{y})) \cup Tail(\vec{y},J(\vec{x},\vec{y}))$
			\item or $\vec{x} = \vec{y}$, in which case $r(g_{\vec{x},b_1}, g_{\vec{y},b_2}) = 0$.
		\end{itemize}
		Therefore
		\begin{align} 
			\sum_{g_{\vec{x},b_1}, g_{\vec{y},b_2} \in \mathcal{Z} } r(g_{\vec{x},b_1},g_{\vec{y},b_2}) \mathbbm{1}_{\{g_{\vec{x},b_1}(v) \ne g_{\vec{y},b_2}(v)\}}
			&\le \sum_{\substack{g_{\vec{x},b_1}, g_{\vec{y},b_2} \in \mathcal{Z}:\\ v \in Tail(\vec{x}, J(\vec{x},\vec{y})) \cup Tail(\vec{y},J(\vec{x},\vec{y}))}} r(g_{\vec{x},b_1},g_{\vec{y},b_2}) \notag \\
			&\le \sum_{\substack{g_{\vec{x},b_1},g_{\vec{y},b_2} \in \mathcal{Z}: \\
					v \in Tail(\vec{x},J(\vec{x},\vec{y}))}} r(g_{\vec{x},b_1},g_{\vec{y},b_2}) 
			+ 
			\sum_{\substack{g_{\vec{x},b_1},g_{\vec{y},b_2} \in \mathcal{Z}:  \\
					v \in Tail(\vec{y},J(\vec{x},\vec{y}))}} r(g_{\vec{x},b_1},g_{\vec{y},b_2})  \notag \\
			&= 2\sum_{\substack{g_{\vec{x},b_1},g_{\vec{y},b_2} \in \mathcal{Z}:  \\
					v \in Tail(\vec{y},J(\vec{x},\vec{y}))}} r(g_{\vec{x},b_1},g_{\vec{y},b_2}) \explain{By symmetry of $r$.}
		\end{align}
		This completes the proof of the lemma.
	\end{proof}
	
	\begin{lemma} \label{lem:exists_Z_with_qZ_positive}
		In the setting of Theorem~\ref{thm:lower_bound_in_terms_of_mixing_time}, if $n \ge 16\sigma^2$ then there exists a subset  $\cZ \subseteq \mathcal{X}$  with $q(\cZ) > 0$.
	\end{lemma}
	\begin{proof}
		By definition of $\sigma$ and because $\max_{v \in V} \pi(v) \ge 1/n$, we have
		\begin{align} \label{eq:min_pi}
			\min_v \pi(v) \ge 4/n^{1.5} \,.
		\end{align}
		
		Fix an arbitrary vertex $u$. Let $S$ be the set of vertices unreachable from $u$ via a random walk   that evolves according to $\cP$  and has at most $T = t_{mix}(\frac{\sigma}{2n})$ steps. Denote $s = |S|$. Then 
		\begin{align} \label{eq:unreachable_from_u_PT}
			(\cP^T)_{u,v} = 0 \; \; \forall v \in S \,.
		\end{align}
		
		By definition of $t_{mix}$, we have
		\begin{align}
			\frac{\sigma}{2n} &\ge \frac{1}{2} \sum_{v \in V} |(\cP^T)_{u,v} - \pi(v)|\notag \\
			& \geq \frac{1}{2} \sum_{v \in S} | \pi(v) | \explain{By \cref{eq:unreachable_from_u_PT}} \\
			&\ge \frac{s}{2} \min_{v \in V} \pi(v) \explain{Since $s = |S|$.}\\
			&\ge \frac{2s}{n^{1.5}} \,. \explain{By \cref{eq:min_pi}.}\\ \label{eq:exists_Z_1}
		\end{align} 
		Meanwhile, since $n \ge 16\sigma^2$, we have
		\begin{align} \label{eq:exists_Z_2}
			\frac{\sigma}{2n} \le \frac{1}{8\sqrt{n}}\,.
		\end{align}
		Combining \cref{eq:exists_Z_1} and \cref{eq:exists_Z_2}, we get $s \le n/16$. Thus the number of vertices reachable  from $u$ in $T$ steps is $n - s \geq 15n/16$.
		For  $n \geq 5$, we have  $15n/16 \geq \lfloor \sqrt{n} \rfloor + 2$, which means that at least $\lfloor \sqrt{n} \rfloor + 2$ vertices are reachable via $\cP$ from any vertex $u$ within $T$ steps.
		
		We can then construct two walks, $\vec{x}$ and $\vec{y}$, in the following manner.
		For $i = 1, \ldots, \lfloor \sqrt{n} \rfloor$, take $x_{iT}$ to be an arbitrary vertex reachable from $x_{(i-1)T}$ other than $x_0, \ldots, x_{(i-1)T}$; this is possible since at least $\lfloor \sqrt{n} \rfloor + 2$ vertices are reachable. Connect the milestones using an arbitrary path such that every edge $(u,w)$ in the path has $\cP_{u,w} > 0$.
		Construct $\vec{y}$ in the same manner but requiring that $Head(\vec{y}, \lfloor \sqrt{n} \rfloor -1) = Head(\vec{x}, \lfloor \sqrt{n} \rfloor -1)$  and   $y_L \not \in \left\{  x_0, x_T, x_{2T}, \ldots, x_{\lfloor \sqrt{n} \rfloor T} \right\}$.
		
		Define  $\cZ = \{g_{\vec{x}, 0}, g_{\vec{y}, 1}\}$. We will show that $q(\cZ) > 0$. First observe that $\vec{x} \neq \vec{y}$, both $\vec{x}$ and $\vec{y}$ are \emph{good} since they do not repeat milestones, and the bit hidden by $g_{\vec{x},0}$ is different from the bit hidden by $g_{\vec{y},1}$. The length of the prefix shared by $\vec{x}$ and $\vec{y}$ is  $J(\vec{x},\vec{y}) = \lfloor \sqrt{n} \rfloor -1$. Then 
		\begin{align}
			r(g_{\vec{x},0}, g_{\vec{y},1}) = \frac{\cP[\vec{x}] \cP[\vec{y}]}{\cP[Head(\vec{y},\lfloor \sqrt{n} \rfloor -1)]}  > 0\,.\label{eq:positive_r_gx0_gy1}
		\end{align}
		We have 
		\begin{align}
			q(\cZ) & =  \max_{v \in V} \sum_{F_1 \in \cZ} \sum_{F_2 \in \cZ} r(F_1, F_2) \cdot \mathbbm{1}_{\{F_1(v) \ne F_2(v)\}} \notag \\
			& =  2 \max_{v \in V} r(g_{\vec{x}, 0}, g_{\vec{y}, 1}) \cdot \mathbbm{1}_{\{g_{\vec{x}, 0}(v) \ne g_{\vec{y}, 1}(v)\}} \explain{Since $\cZ = \{g_{\vec{x}, 0}, g_{\vec{y}, 1}\}$ and $r(F, F) = 0$  $\forall F \in \cZ$.} \\
			& > 0\,.  \explain{Using \cref{eq:positive_r_gx0_gy1} and  $g_{\vec{x},0}(x_L) \neq g_{\vec{y},1}(y_L)$.}
		\end{align}
		This completes the proof.
	\end{proof}
	
	\begin{lemma} \label{lem:derivative}
		For  all $x \ge 2y \ge 1$ we have
		\begin{align}
			\frac{\partial}{\partial x} \left( 1 - \frac{y}{x}\right)^x \ge 0\,.
		\end{align}
	\end{lemma}
	\begin{proof}
		
		Define $z = x/y$.
		Then
		\begin{align}
			\frac{\partial}{\partial x} \left( 1 - \frac{y}{x}\right)^x
			&= \frac{\partial z}{\partial x} \frac{\partial}{\partial z} \left( \left( 1 - \frac{1}{z}\right)^z \right)^y \explain{By chain rule.} \\
			&= \frac{y}{y} \left( \left( 1 - \frac{1}{z}\right)^z \right)^{y-1} \frac{\partial}{\partial z} \left( 1 - \frac{1}{z}\right)^z \explain{By chain rule.} \\
			&\ge 0 \,. \explain{Since $z > 1$}
		\end{align}
		This concludes the proof of the lemma.
	\end{proof}
	
	\begin{lemma}[Folklore] \label{lem:reversibility_markov_chain}
		Consider a reversible Markov chain on $G$ with transition matrix $\cP$. For all $u,v \in V$ and $\ell \in \mathbbm{N}$, we have
		\begin{align*}
			\mbox{\em E}_{visit}(u,v,\ell) \pi(u) = \mbox{\em E}_{visit}(v,u,\ell) \pi(v)\,.
		\end{align*}
	\end{lemma}
	\begin{proof}
		We have
		\begin{align}
			\mbox{E}_{visit}(u,v,\ell) \pi(u) &= \sum_{i=1}^\ell \pi(u) \cP^i_{u,v} \explain{By definition of $\mbox{E}_{visit}$.}\\
			&= \sum_{i=1}^\ell \pi(v) \cP^i_{v,u} \explain{By \cite{levin2017markov} equation 1.30.}\\
			&= \mbox{E}_{visit}(v,u,\ell) \pi(v) \,. \explain{By definition of $\mbox{E}_{visit}$.}
		\end{align}
		This concludes the proof of the lemma.
	\end{proof}
	
	We also use  the following lemma from \cite{BCR23}.
	
	\begin{lemma}[\cite{BCR23}, Lemma 6] \label{lem:valid_unique_local_min}
		Suppose $\vec{x} = (x_0, x_1, \ldots, x_\ell)$ is a walk on $G$ and $f : V \to \mathbbm{R}$ is a valid function for the walk $\vec{x}$. Then $f$ has a unique local minimum at $x_\ell$, the last vertex on the walk.
	\end{lemma}
	

	\section{Corollaries of the main theorem}
	
	
	We can connect \cref{thm:lower_bound_in_terms_of_mixing_time} to the spectral gap of the transition matrix of the Markov chain used
	via the following inequality~
	(see, e.g., \cite{levin2017markov}, Theorem 12.4):
	If $\cP$ is lazy, irreducible, and time-reversible, then
	\begin{align} \label{eq:mixing_time_and_spectral_gap}
		t_{mix}(\epsilon) \leq \left( \frac{1}{1-\lambda_2} \right) \log \left( \frac{1}{\epsilon  \min_{x \in V} \pi(x)}
		\right) \,.
	\end{align}
	
	We obtain the following corrollary, which lower bounds the randomized complexity of local search as a function of the spectral gap of $\cP$.
	
	\highspectralgapimplieslocalsearchhard*
	\begin{proof}
		We proceed by substituting \cref{eq:mixing_time_and_spectral_gap} into \cref{thm:lower_bound_in_terms_of_mixing_time}.
		This directly yields that the randomized query complexity of local search on $G$ is
		\begin{align} \label{eq:local_search_and_spectral_gap_unsimplified}
			\Omega\left( \frac{(1-\lambda_2) \sqrt{n}}{\log\big(2n/(\sigma \min_{x \in V} \pi(x))\big) \exp(3\sigma)} \right)\,.
		\end{align}
		By definition of $\sigma$ we have
		\begin{align} \label{eq:sigma_min_pi}
			\sigma \min_{x \in V} \pi(x) = \max_{x \in V} \pi(x) \ge 1/n\,.
		\end{align}
		Combining \cref{eq:local_search_and_spectral_gap_unsimplified} with \cref{eq:sigma_min_pi} yields that the randomized query complexity of local search on $G$ is
		\begin{align}
			\Omega\left( \frac{(1-\lambda_2) \sqrt{n}}{\log(n) \exp(3\sigma)} \right)\,.
		\end{align}
		This completes the proof of the corollary.
	\end{proof}
	
	
	The prior work in \cite{BCR23} implies a result similar to but weaker than \cref{cor:local_search_and_spectral_gap}.
	To get a result in terms of spectral gap via \cite{BCR23}, we need to connect the edge expansion $\beta$ to the spectral gap $1 - \lambda_2$ of a particular Markov chain.
	We will do this via the bottleneck ratio $\Phi_\star$.
	
	\begin{lemma}[\cite{levin2017markov}, Theorem 13.3] \label{lem:mixing_time_and_bottleneck_ratio}
		If the Markov chain is lazy, then 
		\begin{align}
			\frac{\Phi_\star^2}{2} \le 1 - \lambda_2 \le 2\Phi_\star \,.
		\end{align}
	\end{lemma}
	
	We can use this to get a bound on local search in terms of spectral gap via the following lemma from \cite{BCR23}.
	\begin{lemma}[\cite{BCR23}, corollary 2] \label{lem:local_search_and_expansion}
		The randomized query complexity of local search is in
		\begin{align}
			\Omega\left( \frac{\beta \sqrt{n}}{d_{max} \log^2(n)} \right)
		\end{align}
	\end{lemma}
	
	As a special case, consider the simple lazy random walk on a graph with $d_{max}/d_{min} \le C$ for constant $C$.
	Applying \cref{lem:local_search_and_expansion} to this walk yields the following result:
	
	\bcrcor*
	
	
	\begin{proof}		
		For
		the simple lazy random walk we have $\pi(u) = d(u)/(2|E|)$ for all $u \in V$.
		Let $\cP$ be the transition matrix of the simple lazy random walk.
		Then
		\begin{align}
			\Phi_{\star} &= \min_{S \subseteq V \;\mid\; \pi(S) \le 1/2} \frac{\sum_{(u,v) \in E(S,S^c)} \pi(u) \cP_{u,v}}{\pi(S)} \explain{By definition of $\Phi_{\star}$.}\\
			&= \min_{S \subseteq V \;\mid\; \pi(S) \le 1/2} \frac{\sum_{(u,v) \in E(S,S^c)} (d(u)/(2|E|)) (1/d(u)))}{\sum_{u \in S} (d(u)/(2|E|))} \explain{By definition of $\cP$.} \\
			&= \min_{S \subseteq V \;\mid\; \pi(S) \le 1/2} \frac{|E(S,S^c)|}{\sum_{u \in S} d(u)} \,. \label{eq:phi_beta_1}
		\end{align}
		To get a term of $\beta$, we need to change the scope of $S$ from 
		\[ \{S \subseteq V \mid \pi(S) \le 1/2\}
		\] 
		to 
		\[ 
		\{S \subseteq V \mid |S| \le n/2\}\,.
		\]
		Let $S^*$ be a minimizing choice of $S$ from $\{S \subseteq V \mid |S| \le n/2\}$ for $|E(S,S^c)|/\sum_{u \in S} d(u)$. Then either $S^*$ or $S^{*c}$ (or both) will be in $\{S \subseteq V \mid \pi(S) \le 1/2\}$. We analyze these two cases separately.
		
		\paragraph{Case 1: $\pi(S^*) \le 1/2$.}
		Continuing from \cref{eq:phi_beta_1}, this gives us
		\begin{align}
			\Phi_{\star} &\le \frac{|E(S^*, S^{*c})|}{\sum_{u \in S^*} d(u)} \explain{Since $\pi(S^*) \le 1/2$.} \\
			&= \min_{S \subseteq V \;\mid\; |S| \le n/2} \frac{|E(S,S^c)|}{\sum_{u \in S} d(u)} \,. \explain{By definition of $S^*$.} \\ \label{eq:phi_beta_case_1}
		\end{align}
		
		\paragraph{Case 2: $\pi(S^{*c}) \le 1/2$.}
		Continuing from \cref{eq:phi_beta_1}, this gives us
		\begin{align}
			\Phi_{\star} &\le \frac{|E(S^*, S^{*c})|}{\sum_{u \in S^{*c}} d(u)} \explain{Since $\pi(S^{*c}) \le 1/2$.}\\
			&\le \frac{|E(S^*, S^{*c})|}{\sum_{u \in S^*} d(u)} \cdot C \explain{Since $|S^*| \le n/2 \le |S^{*c}|$ and by definition of $C$.}\\
			&= \min_{S \subseteq V \;\mid\; |S| \le n/2} \frac{|E(S,S^c)|}{\sum_{u \in S} d(u)} \cdot C \,. \explain{By definition of $S^*$.} \\ \label{eq:phi_beta_case_2}
		\end{align}
		
		In both cases, we have
		\begin{align}
			\Phi_{\star}
			&\le \min_{S \subseteq V \;\mid\; |S| \le n/2} \frac{|E(S,S^c)|}{\sum_{u \in S} d(u)} \cdot C \explain{By \cref{eq:phi_beta_case_1} and \cref{eq:phi_beta_case_2}} \\
			&\le \min_{S \subseteq V \;\mid\; |S| \le n/2} \frac{|E(S,S^c)|}{|S| \cdot d_{min}} \cdot C \explain{By definition of $d_{min}$.} \\
			&= \frac{\beta \cdot C}{d_{min}} \explain{By definition of $\beta$.}\\
			&\le \frac{\beta \cdot C^2}{d_{max}}\,. \explain{Since $d_{max}/d_{min} \le C$.} \\ \label{eq:phi_beta}
		\end{align}
		
		Substituting \cref{eq:phi_beta} into \cref{lem:mixing_time_and_bottleneck_ratio}, we get
		\begin{align} \label{eq:beta_and_spectral_gap}
			1 - \lambda_2 \le \frac{2\beta \cdot C^2}{d_{max}} \,. 
		\end{align}
		Substituting \cref{eq:beta_and_spectral_gap} into \cref{lem:local_search_and_expansion}, we get a lower bound of $  \Omega\bigl( \frac{(1 - \lambda_2) \sqrt{n}}{\log^2(n)} \bigr)$ on  the randomized query complexity of local search on $G$.
		This completes the proof of the corollary.
	\end{proof}
	
	Compare to the following corollary, which is just \cref{cor:local_search_and_spectral_gap} applied to the simple lazy random walk when $d_{max}/d_{min}$ is bounded by a constant $C$.
	
	\highspectralgapimplieslocalsearchhardboundedsigma*
	\begin{proof}
		Since $d_{max}/d_{min} \le C$, we have that $\exp(3\sigma) \le \exp(3C)$, which is a constant.
		Therefore \cref{cor:local_search_and_spectral_gap} directly gives that the randomized query complexity of local search is
		\begin{align}
			\Omega\left( \frac{(1-\lambda_2) \sqrt{n}}{\log(n) \exp(3\sigma)} \right) = \Omega\left( \frac{(1-\lambda_2) \sqrt{n}}{\log(n)} \right)\,.
		\end{align}
	\end{proof}
	
	\cref{cor:local_search_and_spectral_gap_bounded_sigma} is stronger than \cref{cor:local_search_and_spectral_gap_sharp_power_law} by a factor of $\log(n)$.
	This represents an improvement of this paper in bounding the difficulty of local search in terms of spectral gap.

\end{document}